\documentclass[]{llncs}
\usepackage[utf8]{inputenc}
\usepackage[T1]{fontenc}
\usepackage{graphicx}
\usepackage{grffile}
\usepackage{longtable}
\usepackage{wrapfig}
\usepackage{rotating}
\usepackage[normalem]{ulem}
\usepackage{amsmath}
\usepackage{textcomp}
\usepackage{amssymb}
\usepackage{capt-of}
\usepackage{hyperref}
\usepackage{tikz}
\author{Matthieu Lemerre \and Sébastien Bardin}
\usepackage{mathrsfs}
\usepackage{mathpartir}
\usepackage{mdwlist}
\usepackage{stmaryrd} 
\usepackage{adjustbox}  
\usepackage{array} 
\usepackage[numbers]{natbib} 
\usetikzlibrary{shapes.geometric}
\usetikzlibrary{calc} \usetikzlibrary{decorations.markings}
\institute{CEA, LIST}
\usepackage{float}
\floatstyle{boxed}
\restylefloat{figure}
\date{\today}
\title{Abstract Interpretation using a Language of Symbolic Approximation}
\hypersetup{
 pdfauthor={Matthieu Lemerre},
 pdftitle={Abstract Interpretation using a Language of Symbolic Approximation},
 pdfkeywords={},
 pdfsubject={},
 pdfcreator={Emacs 25.1.1 (Org mode 9.1.4)}, 
 pdflang={English}}
\begin{document}

\maketitle
\renewcommand{\topfraction}{0.98}	
\renewcommand{\bottomfraction}{0.98}	
    \setcounter{topnumber}{2}
    \setcounter{bottomnumber}{2}
    \setcounter{totalnumber}{4}     
    \setcounter{dbltopnumber}{2}    
    \renewcommand{\dbltopfraction}{0.9}	
    \renewcommand{\textfraction}{0.02}	
    \renewcommand{\floatpagefraction}{0.9}	
    \renewcommand{\dblfloatpagefraction}{0.9}	

\setlength{\textfloatsep}{3.3ex}
\setlength{\intextsep}{3.3ex}

\begin{abstract}  

The traditional abstract domain framework for imperative programs suffers from several shortcomings;
in particular it does not allow precise symbolic abstractions.  
To solve these problems, we propose a new abstract interpretation framework, based on symbolic expressions used both as an \emph{abstraction} of the program, and as the \emph{input} analyzed by abstract domains.
We demonstrate new applications of the framework: 
 an abstract domain that efficiently propagates constraints \emph{across the whole program}; 
 a new formalization of functor domains as \emph{approximate translation}, which allows the production of \emph{approximate programs}, on which we can perform classical symbolic techniques.
We used these to build a complete analyzer for embedded C programs, that demonstrates the practical applicability of the framework.
\end{abstract}
\newcommand{\todo}[1]{\textcolor{red}{\\TODO #1}}
\newcommand{\maybe}[1]{\textcolor{orange}{\\MAYBE #1}}
\newcommand{\done}[1]{\textcolor{green}{\\DONE #1}}
\renewcommand{\todo}[1]{}
\renewcommand{\maybe}[1]{}
\renewcommand{\done}[1]{}
\newcommand{\notnow}[1]{}
\newcommand{\muterm}[4]{(\mu #1.\ #2[#3])(#4)}

\section{Introduction}
\label{sec:org48d0c05}

\noindent {\bf Context} The usual \emph{lattice-based} structuring of abstract interpreters for
imperative programs by \citet{cousot1977abstract} consists in
associating, to each program point, the element of a lattice
representing the over-approximation of the set of possible states at
that point. The interface of abstract domains then consists in
transfer functions that interpret the syntax of the program to compute
the lattice element corresponding to the next program point; as well
as lattice operations such as inclusion (\(\sqsubseteq\)), join
(\(\sqcup\)), and widening (\(\nabla\)). This approach has been highly
successful in both laying the theoretical foundation of software
analysis techniques, and applying them in industrial applications
\cite{blanchet2003static,fahndrich2011static,kirchner2015framac}. 
\looseness=-1

\medskip
\noindent {\bf Problem}
Despite these successes the design of industrial-strength analyzers is
still a challenging art. We highlight two problems, detailed in
Section \ref{sec:main-ideas}.

First, the abstract state at each program point contains a mapping of
the whole memory; yet the abstract state between nearby program
points is nearly identical. Consequences include a high memory
consumption and expensive operations. For instance, memory join is
costly, because every memory location is joined, instead
of only the few that differ. Mitigating this problem requires
implementation tricks such as using functional data structures
\cite{blanchet2002design}, possibly with hash-consing
\cite{kirchner2015framac}; essentially the implementation tries to
share what was duplicated by the theoretical framework.

Another important limitation is the handling of symbolic
relations. Symbolic abstract domains are necessary in practical
lattice-based static analyzers
\cite{gange2016uninterpreted,mine2007symbolic,chang2005abstract,djoudi2016recovering,logozzo2008relative,jourdan2015formally},
notably to handle the case where an intermediate computation is put
in a temporary variable (like \texttt{cond} in Figure
\ref{fig:constraint-domain-example}). These abstract domains are
limited by the fact that they must fit in the lattice-based
framework; in particular, their join and widening operations are
always very imprecise.

\medskip
\noindent {\bf Goal, challenge and proposal\ }
Because these limitations are hard to fix in the standard
\emph{lattice-based} framework, our goal is to design an alternative
framework that fixes these issues; but also encompasses the
lattice-based framework, so as to reuse previous work.

We propose a new \emph{term-based} framework, based on a language of
symbolic expressions (Section \ref{sec:laf-syntax-semantics}); which
can be used both as the abstract state inside the abstract domain
(i.e. replacing the lattice elements) or as the input of the abstract
domain (e.g. replacing the program itself). Abstract domains are
functions, that evaluate a symbolic expression to an \emph{abstract value}, that
can be \emph{concretized} into a set of possible values for the variables
of the expression (Section \ref{sec:abstr-evalua}). 

\medskip
\noindent {\bf Contributions}
Our main contribution is the design of a new term-based abstract
interpretation framework. Key ingredients include the Language of
Approximation and Fixpoint (LAF) logic (Section
\ref{sec:laf-collecting-semantics}) with \emph{nondet} and \(\mu\) operators
-- replacing join and widening; its
collecting semantics; and the definition of abstract domains as
abstract interpreter of LAF terms (Section
\ref{sec:abstr-evalua-definition}).

Our second contribution is specific instances of term-based abstract
domains: we show how to lift lattice-based abstract domains to
term-based abstract domains (Section \ref{sec:non-rel-abs-int}) (and
demonstrate improved complexity in the case of non-relational abstract
domains); we provide new abstract domains based on term rewriting
(Section \ref{sec:simple-rewriting-abs-int}). In Section
\ref{sec:constraint-propag-abs-domain} we combine term rewriting and
lattice-based abstract interpretation in a new domain, that performs
backward and forward constraint propagation across the whole program.

Our last contribution is the evaluation of an early implementation of
the approach (Section \ref{sec:complete-example-abs-inter}), but that
already works on industrial case studies and large SVComp
benchmarks. Term-based abstract interpretation allows to
decompose the analyzer as a succession of transformation over LAF
terms, each individually simple, but that mutually refine one another
to provide a precise result. A first experiment demonstrates the
interest of whole-program constraint propagation. A second shows how
we leverage the production of approximate LAF terms by the term-based
abstract domains, to export simplified formula to a Horn-based model
checker.

\section{Motivation and key ideas}
\label{sec:org02a4b4e}
\label{sec:main-ideas}

\maybe{Peut-etre avoir une approche top-down: commencer par l'exemple motivant (hierarchy of translating domains), avant de dire les besoins que ca cree (produire et analyser des termes)}

Our method can be summarized as using a special language of symbolic
expressions; used both as the abstract state inside an abstract
domain, or as the program that the abstract domain analyze. We first
cover the challenges of maintaining a precise symbolic abstraction,
and why we require a new framework for doing so; the benefits of
performing abstract interpretation over these symbolic expressions;
and the organization of an analyzer structured using "translator
abstract domains", both inputting and outputting symbolic expressions.

Figure \ref{fig:constraint-domain-example} illustrates the use of our
language as an abstraction of the program (in the middle), and as the
language on which abstract interpretation is performed (on the
right). This example is explained in detail in Section
\ref{sec:constraint-propag-abs-domain}.

\begin{figure}[t]
\vspace{-3mm}
\begin{minipage}{0.35\textwidth}
\begin{verbatim}
void main(int x){
  int abs,nabs;
  bool cond = x < 0;
  if( cond) 
    { abs = -x;
      nabs = x; }
  else 
    { abs = x;
      nabs = -x; }

  assert(abs == -nabs);
  if(!(abs <= 8)) 
    while(1);
  assert(x/9 == 0);
}
\end{verbatim}
\end{minipage}%
\begin{minipage}{0.31\textwidth}
\begin{align*} 
& \textrm{let\ } c_1 \triangleq x < 0 \\
& \textrm{let\ } nx \triangleq - x \\
& \textrm{let\ } t_1 \triangleq \langle nx, x \rangle \\
& \textrm{let\ } t_1' \triangleq \textrm{assume}(c_1,t_1) \\
& \textrm{let\ } t_2 \triangleq \langle x, nx \rangle \\
& \textrm{let\ } t_2' \triangleq \textrm{assume}(\lnot c_1,t_2) \\
& \textrm{let\ } t_3 \triangleq \textrm{nondet}(t_1',t_2') \\
& \textrm{let\ } abs \triangleq t_3[0] \\
& \textrm{let\ } nabs \triangleq t_3[1] \\
\\
& \textrm{let\ } c_3 \triangleq (abs = -nabs) \\
& \textrm{let\ } xdiv \triangleq x / 9 \\
& \textrm{let\ } c_2 \triangleq abs \le 8 \\
& \textrm{let\ } c_4 \triangleq (xdiv = 0) \\
& \textrm{in\ } c_4 \\
\end{align*}%
\end{minipage}%
\begin{minipage}{0.38\textwidth} \scriptsize
\begin{align*} 
& x   \hspace{-4mm}& & \mapsto  \left[\begin{array}{ll} true & \Vdash [{-}\!\infty;{+}\!\infty] \\ c_1 & \Vdash [{-}\!\infty;{-}1] \\ \lnot c_1 & \Vdash [0;{+}\!\infty] \\ c_1 \land c_2 & \Vdash [-8;-1] \\ \lnot c_1 \land c_2 & \Vdash [0;8] \end{array} \right. \\
& nx   \hspace{-4mm}& & \mapsto  \left[\begin{array}{ll} true & \Vdash [{-}\!\infty;{+}\!\infty] \\ c_1 & \Vdash [1;{+}\!\infty] \\ \lnot c_1 & \Vdash \ [{-}\!\infty;0] \\ c_1 \land c_2 & \Vdash [1;8] \\ \lnot c_1 \land c_2 & \Vdash [-8;0] \end{array} \right. \\
& abs   \hspace{-4mm}& & \mapsto  \left[\begin{array}{ll} true\ & \Vdash [0;{+}\!\infty] \\ c_2 & \Vdash [0;8]  \end{array} \right. \\
& nabs   \hspace{-4mm}& & \mapsto  \left[\begin{array}{ll} true & \Vdash [1;{+}\!\infty]  \end{array} \right. \\
& c_1   \hspace{-4mm}& &  \mapsto  \left[\begin{array}{ll} true & \Vdash \{true;false\} \\ c_1 & \Vdash \{true\} \\  \lnot c_1 & \Vdash \{false\}  \end{array} \right. \\
& c_3   \hspace{-4mm}& &  \mapsto  \left[\begin{array}{ll} true & \Vdash \{true;false\} \end{array} \right. \\
& c_2   \hspace{-4mm}& &  \mapsto  \left[\begin{array}{ll} true & \Vdash \{true;false\} \\ c_2 & \Vdash \{true\} \end{array} \right. \\
& xdiv   \hspace{-4mm}& & \mapsto  \left[\begin{array}{ll}  c_2 & \Vdash [0;0] \end{array} \right. \\
& c_4   \hspace{-4mm}& &  \mapsto  \left[\begin{array}{ll} c_2 & \Vdash \{true\} \end{array} \right. \\
\end{align*}%
\end{minipage}
\vspace{-4mm}
\caption{Applying the constraint propagation abstract domain on a C program (left). Middle: term abstraction built by the domain; $x$ represents the value of {\tt x}. Right: mapping from variables to abstract values, according to some conditions. } 
\label{fig:constraint-domain-example}

\end{figure}

\subsection{Computing symbolic abstractions}
\label{sec:org7cc5184}

Symbolic abstract domains compute abstractions represented by a \emph{term}
(expression) in some language. Existing symbolic abstract domains 
lose precision when assignment, join or widening is performed
\cite{mine2007symbolic,gange2016uninterpreted,djoudi2016recovering}. Indeed, having a fully
precise symbolic abstract domain cause major difficulties, that cannot
be solved using the standard interface to abstract domains.

\smallskip
\noindent {\bf Variables for values, not for
locations\ } In usual A.I. of imperative programs, a precision loss
may occur when a variable \texttt{x} is overwritten: relations such as \(y =
2 * x\) have to be "killed" (if x cannot be substituted). This affects
relational abstract domains \cite{karr1976affine}, including symbolic
domains
\cite{mine2007symbolic,gange2016uninterpreted,djoudi2016recovering}.

But this issue only occurs because variables represent (mutable)
memory \emph{locations}. When variables represent (immutable) \emph{values},
formula remain valid regardless of which memory cell gets
overwritten. In LAF (but also in \cite{chang2005abstract,chang2013modular}), \emph{variables denote values}, thus there is no need
to kill any relation.

\smallskip 
\noindent {\bf Symbolic join is not a least upper bound\ } 
Usual A.I. requires abstract domains to have a lattice structure. We
believe that this requirement does not fit well symbolic abstraction,
and explain why below.

Intuitively, two expressions \(e_1\) and \(e_2\) can be "joined" using a
new expression "\(\textrm{nondet}(e_1,e_2)\)", representing a
non-deterministic choice between \(e_1\) and \(e_2\). But to avoid loosing
precision, the resulting expression should be given a name. For
instance, "\(\textrm{nondet}(2,7) - \textrm{nondet}(2,7)\)" represents
any value in \(\{-5,0,5\}\), as both "\(\textrm{nondet}(2,7)\)"
expressions may evaluate to different values; while "\(\textrm{let\ } v
= \textrm{nondet}(2,7) \textrm{\ in\ } v - v\)" always evaluate
to 0. In LAF (Section \ref{sec:laf-syntax-semantics}), we give a fresh
name to every sub-expression, even if only the result of
nondeterministic operations do require a name.

However, the necessity of giving new unique names to the result of
joining expressions mean that this "symbolic expression abstraction"
cannot have a lattice structure: each "join" of two terms generates a
\emph{different} least upper-bound of these two expressions.

We solve this issue simply by not requiring abstract domains to have a
lattice structure. This choice allows us to handle join (and loops)
precisely while existing lattice-based domains
\cite{gange2016uninterpreted,mine2007symbolic,chang2005abstract,djoudi2016recovering,logozzo2008relative,jourdan2015formally}
cannot. While this is very unusual, abstract interpretation is not
necessarily lattice-based \cite{cousot1992abstract}: the core of
abstract interpretation is to compute an abstract state which is sound
with regards to the collecting semantics of the program, which is what
our framework does (Section \ref{sec:abstr-evalua}). \looseness=-1

\smallskip 
\noindent {\bf Widening does not fit symbolic abstract domains\ } 
LAF introduces a \(\mu\) operator that allows to \emph{fold} expressions,
which is necessary to have a finite abstraction in the presence of
loops. For instance, the contents of \texttt{x} in the program \texttt{while(*) x =
x + x;} is represented in LAF by "\((\mu x. x + x)(x_0)\)", where
\(x_0\) represents the initial value of \texttt{x}. This expression represents a
non-deterministic choice between "\(x_0\)", "\(x_0 + x_0\)",
"\(\mathrm{let\ } x_1 = (x_0 + x_0) \mathrm{\ in\ } x_1 + x_1\)", etc.

However, the standard operator used to find a finite abstraction in
the presence of loops, widening, is not the right tool to find this
folding. In essence, widening amounts to guessing a loop invariant
given (a computation using) the first iterations of a loop. In the
example above, the initial value of \texttt{x} is \(x_0\), and the one at the
beginning of the second iteration is \(x_0 + x_0\). But there is an
infinity of possible foldings that match these two iterations,
including \((\mu x. x + x_0)(x_0)\); \((\mu x. x_0 + x_0 + 37 * (x -
x_0))\); etc.; and this is even more difficult if expressions can be
rewritten before widening!

The interface that we propose for fixpoint computation, based on the
evaluation of LAF terms, takes as input the value at the loop entry
(the argument of the \(\mu\) expression), together with the \emph{effect} of
the loop (the body of the \(\mu\) expression). This way, we always
succeed in computing the folding of an expression, in a single
step. Moreover, it is sufficiently general to also allow widening-like
fixpoint computations; but also acceleration
\citep{gonnord2013abstract}, policy iteration
\citep{costan2005policy}, inductive invariant generation
\cite{oliveira2016polynomial}\ldots{}

\subsection{Benefits of abstract interpretation over symbolic expressions}
\label{sec:orgf6faea2}

LAF terms are purely functional expressions: the values to which they
can evaluate is independent of a program point. Consequently, there is
no need to split semantic information by program points as in
classical analysis. Instead, semantic information can be put in a
\emph{single store}, which enables more sharing. 

For instance, in Figure \ref{fig:constraint-domain-example}, the
semantic store (a map from LAF variables to their values, according to
conditions of the program) is represented at the right of the
figure. The possible values for \texttt{nabs} after the join are centralized
in the store, instead of being duplicated in the abstract element
corresponding to the following program points.

One consequence of this single-store architecture is that transfer
functions are more efficient. For instance the semantic store of
Figure \ref{fig:constraint-domain-example} can be implemented by an
extensible array, and computing the interval of a variable is an
amortized \(\mathcal{O}(1)\) operation. Moreover, \emph{nondet} can be seen
as a join \emph{targeting} the variables that have changed (\texttt{abs} and
\texttt{nabs}). This optimization is essential, as most conditionals impact a
small number of memory locations \cite{blanchet2003static}; targeted
join (and targeted fixpoint) alleviates the need for each abstract
domain to detect the parts of the memory that have changed, improving
the complexity of these operations (Section
\ref{sec:non-rel-abs-int}).

In the usual A.I., refining semantic information between statements is
limited because the semantic information is split into program
points. The single-store architecture removes this
restriction. Section \ref{sec:constraint-propag-abs-domain} details
how, on Figure \ref{fig:constraint-domain-example}, the constraint
propagation traverses the whole program to prove that \texttt{x} \(\in [-8;8]\)
at the end of the program (and the last assertion). The presence of
\emph{nondet} operation and its traversal are essential to relate \texttt{x} in
the program to its absolute value \texttt{abs}. Moreover this constraint
propagation is done using the LAF term, i.e. on the data dependencies
of the program, automatically skipping statements unaffected by the
refinement (like in sparse analyses \cite{oh2012design}).

\maybe{MAYBE Path-sensitive analysis? Est-ce que c'est un exemple? Peut-etre que oui.}
\maybe{MAYBE Indépendence de l'emploi de variables temporaires. Peut-etre en related works}

\subsection{Hierarchy of translator domains}
\label{sec:org2c074f9}
\label{sec:hierarchy-translator-domains}

Our abstract domains are able to both abstractly interpret LAF terms,
and produce a program abstraction as a LAF term. This allows to
implement \emph{translator domains}. Translator domains perform a \emph{dynamic}
translation of an input term into a simplified, possibly approximated,
output term, using semantic information computed during the
analysis (on either the input or output term).

One can then structure an abstract interpreter as a combination of
(simple) translator domains. Superficially, this resembles the
structure of compilers; however translator domains are not just
sequential passes, but abstract domains executing simultaneously, and
mutually refining each other.

\begin{figure}[htbp]
{\begin{minipage}{0\textwidth}

\begin{verbatim}
union { struct { uint8 l; uint8 h; } b;
        uint16 w; } r[3];

r[0].w = X;
if(*) {
  r[1].b.h = r[0].b.h;
  r[1].b.l = r[0].b.l;
  r[1/X].w = r[1].w;
}
assert(r[0].w == X);
\end{verbatim}
\vspace{-3mm}
$\begin{array}{l}
\cr \hspace{0mm}\textrm{let\ } X^l \triangleq X[7..0] \\[-9pt]
\cr \hspace{0mm}\textrm{let\ } X^h \triangleq X[15..8] \\[-9pt]
\cr \hspace{0mm}\textrm{let\ } X' \triangleq \textrm{concat}(X^h,X^l) \\[-9pt]
\cr \hspace{0mm}\textrm{let\ } X'' \triangleq \textrm{nondet}(X,X') \\[-9pt]
\cr \hspace{0mm}\textrm{let\ } X''' \triangleq \textrm{nondet}(X',X') \\[-9pt]
\cr \hspace{0mm}\textrm{let\ } t_{then} \triangleq \langle X'', X'''[7..0],X'''[15..8] \rangle \\[-9pt]
\cr \hspace{0mm}\textrm{let\ } t_{else} \triangleq \langle X, r_1^l,r_1^h \rangle \\[-9pt]
\cr \hspace{0mm}\textrm{let\ } t \triangleq \textrm{nondet}(t_{then},t_{else}) \\[-9pt]
\cr \hspace{0mm}\textrm{let\ } assertion \triangleq (t[0] = X)
\end{array}$

\end{minipage}%
\begin{minipage}{0.5\textwidth}
\vspace{-8.5mm}%
\begin{tikzpicture}

\begin{scope}[shift={(0,0)},xscale=0.9]

  \begin{scope}[yscale=0.5,shift={(-0.5,-0.5)}]
  \draw (0,0) grid (2,1);
  \draw (2,0) grid (4,1);
  \draw (4,0) grid (6,1);
  \end{scope}
  
  \node at (0,0) {$r_0^l$};
  \node at (1,0) {$r_0^h$};
  \node at (2,0) {$r_1^l$};
  \node at (3,0) {$r_1^h$};
  \node at (4,0) {$r_2^l$};
  \node at (5,0) {$r_2^h$};
  \node at (6,0) {$R_0$};
\end{scope}

\begin{scope}[shift={(0,-1)},xscale=0.9]

  \begin{scope}[yscale=0.5,shift={(-0.5,-0.5)}]
  \draw (0,0) rectangle (2,1);
  \draw (2,0) grid (4,1);
  \draw (4,0) grid (6,1);
  \end{scope}
  
  \node at (0.5,0) {$X$};
  \node at (2,0) {$r_1^l$};
  \node at (3,0) {$r_1^h$};
  \node at (4,0) {$r_2^l$};
  \node at (5,0) {$r_2^h$};
  \node at (6,0) {$R_1$};
\end{scope}

\begin{scope}[shift={(0,-2)},xscale=0.9]

  \begin{scope}[yscale=0.5,shift={(-0.5,-0.5)}]
  \draw (0,0) rectangle (2,1);
  \draw (2,0) grid (6,1);
  \end{scope}
  
  \node at (0.5,0) {$X$};
  \node at (2,0) {$r_1^l$};
  \node at (3,0) {$X^h$};
  \node at (4,0) {$r_2^l$};
  \node at (5,0) {$r_2^h$};
  \node at (6,0) {$R_2$};
\end{scope}

\begin{scope}[shift={(0,-3)},xscale=0.9]

  \begin{scope}[yscale=0.5,shift={(-0.5,-0.5)}]
  \draw (0,0) rectangle (2,1);
  \draw (2,0) grid (6,1);
  \end{scope}
  
  \node at (0.5,0) {$X$};
  \node at (2,0) {$X^l$};
  \node at (3,0) {$X^h$};
  \node at (4,0) {$r_2^l$};
  \node at (5,0) {$r_2^h$};
  \node at (6,0) {$R_3$};
\end{scope}

\begin{scope}[shift={(0,-4)},xscale=0.9]

  \begin{scope}[yscale=0.5,shift={(-0.5,-0.5)}]
  \draw (0,0) rectangle (2,1);
  \draw (2,0) rectangle (4,1);
  \draw (4,0) grid (6,1);
  \end{scope}
  
  \node at (0.5,0) {$X''$};
  \node at (2.5,0) {$X'''$};
  \node at (4,0) {$r_2^l$};
  \node at (5,0) {$r_2^h$};
  \node at (6,0) {$R_4$};
\end{scope}

\begin{scope}[shift={(0,-5)},xscale=0.9]

  \begin{scope}[yscale=0.5,shift={(-0.5,-0.5)}]
  \draw (0,0) rectangle (2,1);
  \draw (2,0) grid (4,1);
  \draw (4,0) grid (6,1);
  \end{scope}
  
  \node at (0.5,0) {$t[0]$};
  \node at (2,0) {$t[1]$};
  \node at (3,0) {$t[2]$};
  \node at (4,0) {$r_2^l$};
  \node at (5,0) {$r_2^h$};
  \node at (6,0) {$R_5$};
\end{scope}

\begin{scope}[shift={(0,0.7)}]
\coordinate (p1) at (-0.45,-0.5);
\coordinate (p2) at (-0.7,-0.85);
\coordinate (p3) at (-0.8,-1.7);
\coordinate (p4) at (-0.9,-2.1);
\coordinate (p5) at (-1,-2.4);
\coordinate (p6) at (-1.1,-2.8);
\end{scope}

\draw[black!70] (-0.45,0) -| (p1);
\draw[black!70] (-0.45,-1) -| (p2);
\draw[black!70] (-0.45,-2) -| (p3);
\draw[black!70] (-0.45,-3) -| (p4);
\draw[black!70] (-0.45,-4) -| (p5);
\draw[black!70] (-0.45,-5) -| (p6);

\draw[black!70] (p1) -- +(-4.1,0);
\draw[black!70] (p2) -- +(-3.85,0);
\draw[black!70] (p3) -- +(-2.1,0);
\draw[black!70] (p4) -- +(-2.0,0);
\draw[white] (p5) -- +(-5.4,0);
\draw[black!70] (p5) -- +(-5.2,0);
\draw[black!70] (p6) -- +(-5.1,0);

\end{tikzpicture}
\end{minipage}}%

\caption{A hierarchy of abstract domains handling low-level memory operations. 
The right part contains the abstract representation of the memory region for \texttt{r} at different
program points. This representation contains variables of the LAF term (with bitvector theory) on the left. $\cdot[b..a]$ denotes substring extraction of a bitvector.}
\label{fig:example-mine}

\end{figure}
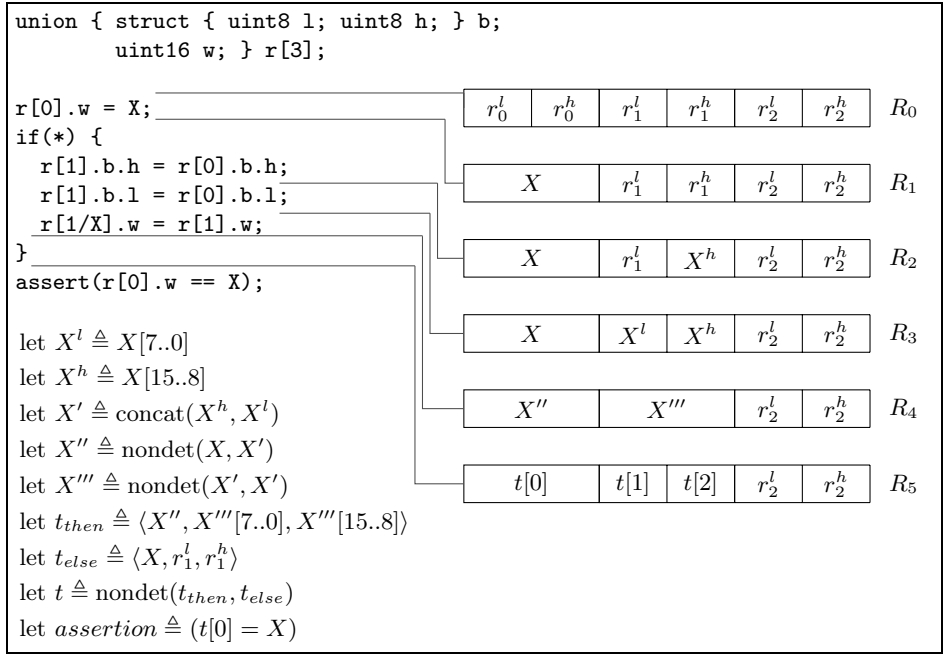

\smallskip
\noindent {\bf Example: translator abstract domain for low-level memory operations\ }
Figure \ref{fig:example-mine} shows how this can be applied to the
static analysis of C programs with low-level memory manipulation such
as pointer arithmetic, casts and bitfield (similarly to
\citet{mine2006field}). The memory abstract domain represents memory
regions as a contiguous sequence of slices. To each slice correspond
the value that was last written. These values are represented as
variables of a LAF term (in the theory of bitvectors). The right of
the figure represents the memory region for \texttt{r} at different program
points. Each write modifies the region by replacing the slice with the
corresponding offset, size, and value.

\maybe{MAYBE: Faire un X''' qui vaudrait nondet(X',X')?}
\maybe{MAYBE: Montrer dans l'exemple Le cas où on scinde les slices adjascents}
\maybe{MAYBE: Fusionner cet exemple avec le premier?}
\todo{TODO: Donner la traduction en terme de regions? Et le schema de l'analyseur  (avec l'export et la simplification des termes)}

The memory domain has to know the possible values for array
indices. This is done by \emph{querying} the underlying abstract domain for
the possible values of these indices. If there is more than one
possible value (e.g. \(1/X \in \{-1;0;1\}\)), the abstract domain
performs a weak update, i.e. the corresponding slice is mapped to a
non-deterministic choice between the new and old value (\(R_4\)). On
program joins, the slices are compared: only the cells that do not
contain syntactically equal variables are "joined", which targets the
join on memory cells that have been modified by the conditional.

\maybe{MAYBE The resulting domain is simple. The main feature that makes its
definition possible is the possibility of "joining" variables.}

\smallskip
\noindent {\bf Proving the assertion with symbolic reasoning\ }
\label{sec:proving-assertion-symbolic-reasoning}
Now, there are several ways to solve the final assertion. The first is
to insert between these domains a syntactic rewrites abstract
domains. Applying the following four simple rewriting rules suffices
to prove the assertion (all primed versions of \(X\) are equal):
\begin{align*}
& concat(x[a..b],x[b-1..c]) \ \to\  x[a..c] & & x[0..c] \textrm{\ when\ sizeof}(x) = c\ \to\  x \\
& \textrm{nondet}(x,x) \ \to\  x & & x == x \ \to\  true \\
\end{align*}
Another mean is to output the corresponding formula to a specialized
SMT solver, which is necessary for the most complex
assertions. Contrary to the original program, the simplified LAF
formula does not require memory operators: it only needs operator from
the bitvector theory, that many SMT solver understand. Thus, the
memory abstract domain can be viewed as a translator from the source
program to a simpler, approximate program, directly suitable for
export to specialized solvers.

\maybe{MAYBE: tagless final; dynamic dataflow}
\todo{Etendre avec la traduction originale (avec les loads). Et la hierarchie de domaines abstrait. Mettre cet exemple en premier}

\section{LAF: syntax and collecting semantics}
\label{sec:org974a88f}
\label{sec:laf-syntax-semantics}

The Logic of Approximation and Fixpoint (LAF) is the language we
designed to be used both as a symbolic abstraction in abstract
domains, and as an input for abstract domains. LAF can be viewed as a
non-deterministic, functional language representing the possible
results of a computation.

\maybe{This section details the syntax and collecting semantics of LAF; the latter is used to define the soundness of abstract interpretation over LAF terms.}

\begin{figure}[htbp]
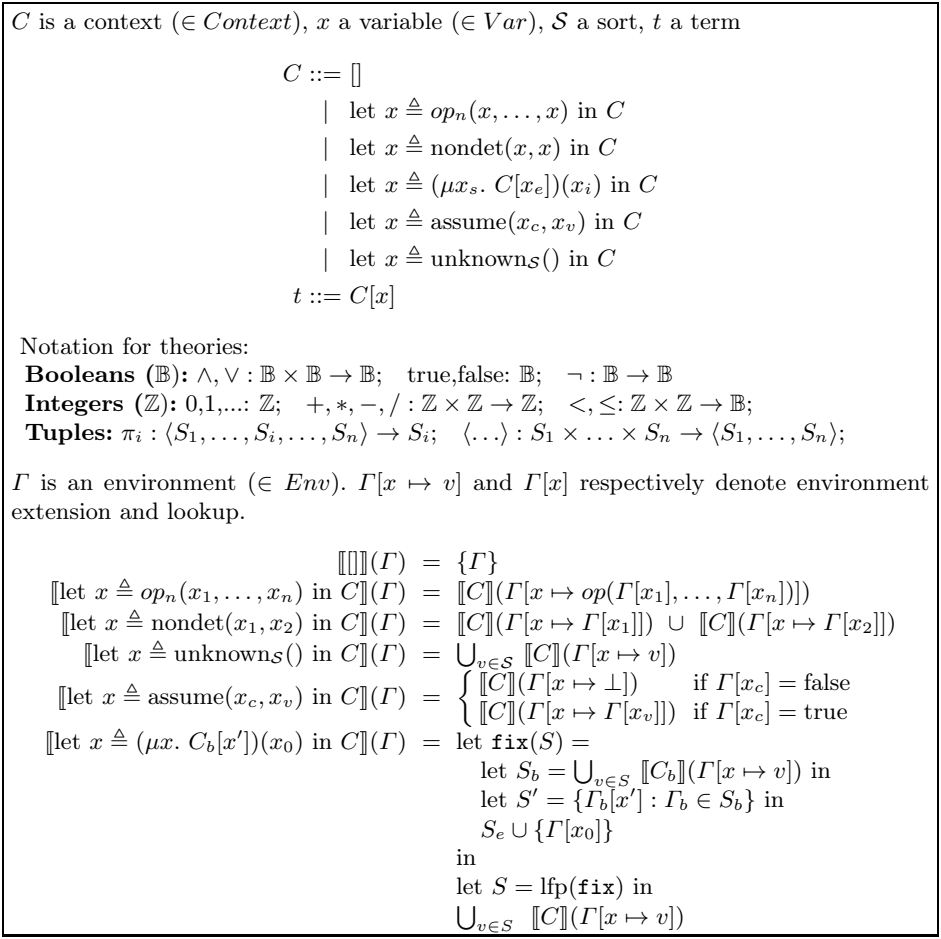

$C$ is a context ($\in Context$), $x$ a variable ($\in Var$), $\mathcal{S}$ a sort, $t$ a term
\begin{eqnarray*}
C & ::= & [] \\
  & | & \textrm{let\ } x \triangleq op_n(x,\ldots,x) \textrm{\ in\ } C \\
  & | & \textrm{let\ } x \triangleq \textrm{nondet}(x,x) \textrm{\ in\ } C \\
  & | & \textrm{let\ } x \triangleq \muterm{x_s}{C}{x_e}{x_i} \textrm{\ in\ } C \\
  & | & \textrm{let\ } x \triangleq \textrm{assume}(x_c,x_v) \textrm{\ in\ } C \\
  & | & \textrm{let\ } x \triangleq \textrm{unknown}_\mathcal{S}() \textrm{\ in\ } C \\
t & ::= & C[x] 
\end{eqnarray*}%
\vspace{-3mm}
Notation for theories:
\begin{description}
\item{\bf Booleans ($\mathbb{B})$:} $\land, \lor: \mathbb{B} \times \mathbb{B} \to \mathbb{B}$;\quad true,false: $\mathbb{B}$;\quad $\lnot: \mathbb{B} \to \mathbb{B}$
\item{\bf Integers ($\mathbb{Z})$:} 0,1,...: $\mathbb{Z}$;\quad $+,*,-,/: \mathbb{Z} \times \mathbb{Z} \to \mathbb{Z}$;\quad $<,\le: \mathbb{Z} \times \mathbb{Z} \to \mathbb{B}$;
\item{\bf Tuples:} $\pi_i: \langle S_1,\ldots,S_i,\ldots,S_n \rangle \to S_i$;\quad $\langle \ldots \rangle: S_1 \times \ldots \times S_n \to \langle S_1,\ldots,S_n \rangle$;
\end{description}
\vspace{-4mm}
\label{fig:laf-syntax}
\vspace{4mm}




$\Gamma$ is an environment ($\in Env$). $\Gamma[x \mapsto v]$ and $\Gamma[x]$ respectively denote environment extension and lookup. 
\[
\begin{array}{r@{\hspace{2mm}}c@{\hspace{2mm}}l}
\llbracket [] \rrbracket(\Gamma) & = & \{ \Gamma \} \\
\llbracket \textrm{let\ } x \triangleq op_n(x_1,\ldots,x_n) \textrm{\ in\ } C \rrbracket(\Gamma) & = & \llbracket C \rrbracket( \Gamma[x \mapsto {op}(\Gamma[x_1],\ldots,\Gamma[x_n])]) \\
\llbracket \textrm{let\ } x \triangleq \textrm{nondet}(x_1,x_2) \textrm{\ in\ } C \rrbracket(\Gamma) & = & \llbracket C \rrbracket( \Gamma[x \mapsto \Gamma[x_1]])\ \cup\ \llbracket C \rrbracket( \Gamma[x \mapsto \Gamma[x_2]]) \\
\llbracket \textrm{let\ } x \triangleq \textrm{unknown}_{\mathcal{S}}() \textrm{\ in\ } C \rrbracket(\Gamma) & = & \bigcup_{v \in \mathcal{S}}\ \llbracket C \rrbracket( \Gamma[x \mapsto v]) \\
\llbracket \textrm{let\ } x \triangleq \textrm{assume}(x_c,x_v) \textrm{\ in\ } C \rrbracket(\Gamma) & = & \left\{ \begin{array}{ll} \llbracket C \rrbracket(\Gamma[x \mapsto \bot]) & \textrm{\ if\ } \Gamma[x_c] = \textrm{false} \\
                                                                                               \llbracket C \rrbracket(\Gamma[x \mapsto \Gamma[x_v]]) & \textrm{\ if\ } \Gamma[x_c] = \textrm{true} \end{array} \right.\\
\llbracket \textrm{let\ } x \triangleq \muterm{x}{C_b}{x'}{x_0} \textrm{\ in\ } C \rrbracket(\Gamma) & = 
& \textrm{let\ \tt fix}(S) = \\
& & \quad \textrm{let\ } S_b = \bigcup_{v \in S} \ \llbracket C_b \rrbracket(\Gamma[x \mapsto v]) \textrm{\ in} \\
& & \quad \textrm{let\ } S' = \{ \Gamma_b[x'] : \Gamma_b \in S_b \} \textrm{\ in} \\
& & \quad S_e \cup \{ \Gamma[x_0] \}\\
& & \textrm{in\ } \\
& & \textrm{let\ } S = \textrm{lfp}(\tt fix) \textrm{\ in\ } \\
& & \bigcup_{v \in S}\ \ \llbracket C \rrbracket(\Gamma[x \mapsto v])
\end{array}
\]
\vspace{-4mm}
\caption{Syntax (top) and collecting semantics (bottom) of LAF}
\label{fig:direct-collecting-semantics}
\end{figure}

\subsubsection{Syntax}
\label{sec:orgec963c4}
\label{sec:laf-syntax}

A LAF term (Figure \ref{fig:laf-syntax}) is essentially a sequence of
\emph{variable} \emph{definitions}, followed by a single variable (the
\emph{result}). A variable is defined as the result of calling a primitive
operation over previously-defined variables: i.e. every intermediary
computation is \emph{named}, and the sequence of definitions is \emph{ordered}.

The terms need to be built incrementally; i.e. we want to extend a
term with new definitions. We represent this formally using evaluation
contexts \(C\) \cite{wright1994syntactic}, i.e. a "term with a
hole". The hole \([]\) can be substituted with a variable \(x\) (to form a
complete term) or with a context \(C_2\) (to form a new context). For
instance, if \(C = "\textrm{let\ } x \triangleq 12; \textrm{let\ } y \triangleq x + 1;
[]"\) then \(C[y]\) is a term evaluating to \(13\); \(C["\textrm{let\ } z \triangleq
x + y \textrm{\ in\ } []"]\) is a context appended with a new
definition (to improve readability, we usually write the latter
\("C[\textrm{let\ } z \triangleq x + y]"\)).

The primitive operations are those of a logical theory, (e.g. integer,
floating point or real arithmetic; array; bitvectors; uninterpreted
functions): there are many LAF languages, that depend on the theories
in use. Two theories are always present: the theory of Boolean
operations, and the theory of tuples.

\maybe{Sauf qu'on ne représente pas les termes. Donc autant C[x]  c'est pas une substitution, alors que $C_1[C_2]$ c'en est une (abus de   notation)}

\maybe{variables have all unique names.}

\maybe{Ma collecting semantics est semi-collecting; j'en ai une autre qui collecte sur celle-ci... Peut-etre que je devrais juste l'appeller "semantics"}

\subsubsection{Collecting Semantics}
\label{sec:orgd23adb5}
\label{sec:laf-collecting-semantics}

Because variable definitions are ordered, we can define an operational
semantics of LAF terms. LAF can be thus seen as a nondeterministic,
functional language. A semantics of LAF given as a small-step
structural operational semantics \cite{plotkin1981structural} exists
(see Appendix \ref{sec:laf-small-step-operational-semantics}). But
giving the collecting semantics for the term directly is actually
simpler (see Figure \ref{fig:direct-collecting-semantics}); this
illustrates the fact that LAF is a logic adequate for a symbolic
description of a collecting semantics.

The semantics is defined using a \emph{collecting evaluation} function
\\\mbox{$\llbracket \cdot \rrbracket: Context \times Env \to \mathcal{P}(Env)$}.
It takes a context \(C \in Context\) remaining to be evaluated, an
\emph{environment} \(\Gamma \in Env\) corresponding to a possible evaluation
so far (more precisely, \(\Gamma\) should contain a mapping for all the
free variables in \(C\)); and returns all the possible corresponding
environments when \(C\) is fully evaluated. To evaluate a closed context
(with no free variables), we pass the empty environment \(\varepsilon\).

The collecting semantics essentially computes the set of all possible
environments using meet-over-all-paths
\cite{nielson1999principles}. For non-deterministic constructs
(\emph{nondet}, \emph{unknown}, and \(\mu\)), every possible choice is fully
evaluated in isolation, before taking the union of the possible
outcomes. \(\mu\) defines a local fixpoint of all the values obtained by
an arbitrary number of iterations of the loop body \(C_b[x']\), before
making a nondeterministic choice of one of these values. In other
words, \(\mu\) is as an operator executing the loop body a
nondeterministic number of times.

Intuitively, "assume(false,x)" is used to "kill" a part of the
evaluation. A particularity of LAF is that the same evaluation can
have some parts killed and some live (e.g. only one of the branch of a
conditional is live). For instance, the desired result of evaluating
the term \("\textrm{let\ } x \triangleq \textrm{assume}(false,1); \textrm{let\ }
y \triangleq 2; y"\) is 2. This is achieved using a special \(\bot\) value to
represent "killed" values. \(\bot\) propagates throughout theory
operations: e.g. \(\bot + 3 = \bot\). 

\maybe{This is unusual: generally such
"killing" is handled by not having a transition from the "evaluation
killing" instruction.  but this would lead to an undesired behavior for
the term "let x = assume(false,1) in let y = 2 in y".}

\maybe{Also note that the collecting semantics is a set, but it is possible to encode traces using a theory of sequences: the framework is not limited to invariant or relational properties.}

\section{Abstract evaluation}
\label{sec:org7d544d7}
\label{sec:abstr-evalua}

\subsection{Definition}
\label{sec:orgb9fc3e9}
\label{sec:abstr-evalua-definition}

\done{Definir avant Context,Env,Var, pour leur rajouter des chapeaux}

We name \emph{abstract evaluation} the abstract interpretation of
functional LAF terms. An abstract domain is mainly composed of an
\emph{abstract evaluation function} \(\llbracket \cdot \rrbracket^\sharp\)
that evaluates contexts into \emph{abstract states}; and a \emph{concretization}
function \(\gamma\) that gives a meaning to an abstract state as a set
of environments. The main soundness rule is that the abstract state is
always a superset of the set of possible environments of the context,
as defined by the collecting semantics.

To allow the incremental building of the abstract value, the abstract
evaluation function takes as an argument the abstract value computed
so far, which overapproximates the set of environments containing the
free variables of the context that remains to be evaluated. More
formally:
\begin{definition}[Abstract Domain] An abstract domain is given by a quadruple $\langle Env^\sharp, \llbracket \cdot \rrbracket^\sharp, \varepsilon^\sharp, \gamma: Env^\sharp \to \mathcal{P}(Env) \rangle$ where:
\vspace{-2.5mm}
\begin{itemize}
\item $Env^\sharp$ is a set, the set of \emph{abstract states} (also called \emph{abstract environments});
\item $\llbracket \cdot \rrbracket^\sharp \in Context \times Env^\sharp \to Env^\sharp$ is the \emph{abstract evaluation function};
\item $\varepsilon^\sharp \in Env^\sharp$ is the \emph{initial abstract state};
\item $\gamma \in Env^\sharp \to \mathcal{P}(Env)$ is the \emph{concretization function}.
\end{itemize}
\label{def:abs-domain}
\end{definition}
\begin{definition}[Soundness of Abstract Domains]
An abstract domain \\ \mbox{$\langle Env^\sharp, \llbracket \cdot \rrbracket^\sharp, \varepsilon^\sharp, \gamma: Env^\sharp \to \mathcal{P}(Env) \rangle$} is \emph{sound} if the following rules apply:\\
\indent 1. Soundness of the initial state: $\varepsilon \in \gamma(\varepsilon^\sharp)$ \\
\indent 2. Soundness of abstract evaluation: \ $\displaystyle \forall C, \Gamma^\sharp: \ \bigcup_{\Gamma \in \gamma(\Gamma^\sharp)} \llbracket C \rrbracket(\Gamma) \subseteq \gamma(\llbracket C \rrbracket^\sharp(\Gamma^\sharp)) $
\label{def:sound-abs-domain}
\end{definition}
Note how these definitions have replaced the traditional lattice-based
structure of abstract domains by the sole evaluation of operators in
the logic. Section \ref{sec:non-rel-abs-int} shows that traditional
abstract domains fit into this definition; while Section
\ref{sec:simple-rewriting-abs-int} shows how the definition allows to
also incorporate techniques not traditionally seen as abstract
domains, for instance term rewriting.

\subsection{Numerical abstract interpretation}
\label{sec:orgb2ee4a1}
\label{sec:non-rel-abs-int}

A non-relational abstract domain is a mapping from variables to an
abstraction of the value of this variable, that we call \emph{abstract
value}. Example of abstract values include intervals, congruences
\cite{granger1989congruences}, the flat lattice of constants, the
powerset of \(\{ \mathrm{true}, \mathrm{false} \}\)\ldots{} They have a
lattice-based structure and can be combined using the operators of the
LAF theories (e.g. \(\dot{+}\) denotes the addition of two intervals).

\begin{definition}[Abstract value] An \emph{abstract value} $\mathcal{V}$ is a pair $\langle L_\mathcal{V}, \gamma_{\mathcal{V}}\rangle$ where:
\\ \textcolor{white}{-} 1. $L_\mathcal{V}$ is a lattice, equipped with join ($\dot\sqcup$), inclusion ($\dot\sqsubseteq$), widening ($\dot\nabla$) operations, as well as abstractions of theory operations $\dot{op}$
\\ \textcolor{white}{-} 2. $\gamma_\mathcal{V}: L_\mathcal{V} \to \mathcal{P}(Values)$ concretizes elements of the lattice into a set of values;
\\ \textcolor{white}{-} 3. For every $n$-ary operation $\dot{op}$ over lattice elements:\\ $\displaystyle \gamma_\mathcal{V}(\dot{op}(L_1,\ldots,L_n)) \supseteq \{ op(x_1,\ldots,x_n) : x_1 \in \gamma_\mathcal{V}(L_1),\ldots, x_n \in \gamma_\mathcal{V}(L_n) \} $
\\ \textcolor{white}{-} 4. $\displaystyle \gamma_\mathcal{V}(L_1 \mathbin{\dot\sqcup} L_2) \supseteq \gamma_\mathcal{V}(L_1) \cup \gamma_\mathcal{V}(L_2) $
\end{definition}

Figure \ref{fig:non-relational-abstract-domain} presents a
non-relational abstract domain equipped with such an abstract value,
i.e. a concretization and an abstract evaluation algorithm which
computes the over-approximation of the possible values of a LAF
term. 

\begin{figure}[htbp]

\[
\begin{array}{l@{\hspace{2mm}}l}
\varepsilon^\sharp = & \varepsilon \\
\llbracket [] \rrbracket^\sharp(\Gamma^\sharp) = & \Gamma^\sharp \\
\llbracket \textrm{let\ } x \triangleq op_n(x_1,\ldots,x_n) \textrm{\ in\ } C \rrbracket^\sharp(\Gamma^\sharp) = & \llbracket C \rrbracket^\sharp( \Gamma^\sharp[x \mapsto \dot{op}(\Gamma^\sharp[x_1],\ldots,\Gamma^\sharp[x_n])]) \\
\llbracket \textrm{let\ } x \triangleq \textrm{nondet}(x_1,x_2) \textrm{\ in\ } C \rrbracket^\sharp(\Gamma^\sharp) = & \llbracket C \rrbracket^\sharp( \Gamma^\sharp[x \mapsto \Gamma^\sharp[x_1] \mathbin{\dot{\sqcup}} \Gamma^\sharp[x_2]]) \\
\llbracket \textrm{let\ } x \triangleq \textrm{unknown}() \textrm{\ in\ } C \rrbracket^\sharp(\Gamma^\sharp) = & \llbracket C \rrbracket^\sharp( \Gamma^\sharp[x \mapsto \dot{\top}]) \\
\llbracket \textrm{let\ } x \triangleq \textrm{assume}(x_c,x_v) \textrm{\ in\ } C \rrbracket^\sharp(\Gamma^\sharp) = & \llbracket C \rrbracket^\sharp( \Gamma^\sharp[x \mapsto \Gamma^\sharp[x_v]]) \\
\llbracket \textrm{let\ } x \triangleq \muterm{x_s}{C_b}{x_e}{x_i} \textrm{\ in\ } C \rrbracket^\sharp(\Gamma^\sharp) =
& \textrm{let\ } L_i = \Gamma^\sharp[x_i] \textrm{\ in} \\
& \textrm{let\ rec\ \tt fix}(L) = \\
& \quad \textrm{let\ } \Gamma^\sharp_s = \Gamma^\sharp[x_s \mapsto L] \textrm{\ in} \\
& \quad \textrm{let\ } \Gamma^\sharp_b = \llbracket C_b \rrbracket^\sharp(\Gamma^\sharp_s) \textrm{\ in} \\
& \quad \textrm{let\ } L' = \Gamma^\sharp_b[x_e] \mathbin{\dot{\sqcup}} L_i \textrm{\ in} \\
& \quad \textrm{if\ } (L' \mathbin{\dot{\sqsubseteq}} L) \textrm{\ then\ } L \textrm{\ else\ } \texttt{fix}(L \dot{\nabla} L') \\
& \textrm{in\ } \\
& \llbracket C \rrbracket^\sharp(\Gamma^\sharp[x \mapsto \texttt{fix}(L_i)])
\end{array}
\]

\[ \Gamma \in \gamma(\Gamma^\sharp) \quad\Leftrightarrow\quad \forall (x \mapsto v) \in \Gamma:\ v \in \gamma_{\mathcal{V}}(\Gamma^\sharp[x]) \]

\[ Env^\sharp = Var \to L_\mathcal{V} \]

\caption{A non-relational abstract domain, parametrized by an abstract value $\mathcal{V}$}
\label{fig:non-relational-abstract-domain}

\end{figure}

\todo{Donner une notion de soundness d'un domaine abstrait, et l'utiliser ici plutot. I.e. dire que le domaine abstrait de la figure est sound.}

\begin{theorem} Given that $\mathcal{V}$ is an abstract value, the quadruple $\langle Env^\sharp, \llbracket \cdot \rrbracket^\sharp, \varepsilon^\sharp, \gamma \rangle$ of Figure \ref{fig:non-relational-abstract-domain} is a sound abstract domain.
\end{theorem}

The algorithm is quite straightforward: the abstract value is a
standard mapping from program variables to an abstract value
representing their possible values. Every evaluation step applies a
lattice operation \({\dot{op}}\), except \emph{nondet} which requires a
join, and \(\mu\) for which we do a local fixpoint computation. \emph{assume}
is ignored (for more precision, we could map \(x\) to \(\bot\) when we
detect that the condition cannot be true); Section \ref{sec:constraint-propag-abs-domain} explains
how the domain can be extended to handle \emph{assume}.

The complexity of this implementation is \emph{optimal}. Indeed, the
\(\Gamma^\sharp\) environment can be implemented using a single array
(using variables as the indices), which means that environment update
and lookup have \(\mathcal{O}(1)\) complexity. If we assume that the
abstract value of tuple values is represented as a tuple of scalar
abstract values, and operations on scalar abstract values
(e.g. interval) have \(\mathcal{O}(1)\) complexity, then all operations
have \(\mathcal{O}(1)\) complexity, except joining and widening tuples
of length \(n\), which have \(\mathcal{O}(n)\) complexity. Now, the length
of the tuple depends on how it was generated; but if it was generated
so that it contains only the variables that differ, then the
complexity of these operations is \(\mathcal{O}(\Delta)\), where
\(\Delta\) is the number of variables modified in a loop or in a
conditional. This improves on the complexity of analyzers using the
traditional interface, for which the complexity of these operations
depends (logarithmically) on the number of variables
\cite{blanchet2002design}.

\begin{theorem} It is possible to lift traditional, lattice-based relational abstract domains to sound abstract domains following Definition \ref{def:abs-domain}.
\end{theorem}

With minor modifications, the non-relational abstract domain of the
previous section can be used to lift traditional abstract domains to
our framework: it suffices to map memory variables to the usual
lattices representing the entire memory. For instance, we can evaluate
the definition \\{$"\textrm{let\ } M' = \textrm{store}(M, \texttt{"x"}, 
\textrm{load}(M,\texttt{"x"}) + 1)"$\ } using the affine equality abstract
domain \cite{karr1976affine}, associating elements of this domain to
\(M\) and \(M'\). 

It is also possible to use traditional abstract domains to relate
variables of a (numeric) LAF term; Appendix
\ref{sec:lifting-traditional-relational-abstract-domains} details how
this lifting is done. It follows the same pattern as the
non-relational lift: evaluation of "nondet" corresponds to join, and
evaluation of \(\mu\) corresponds to inclusion testing and widening.

\subsection{Rewriting-based abstract interpretation}
\label{sec:org455f158}
\label{sec:simple-rewriting-abs-int}

This section provides a simple example of an abstract domain whose
abstract state is based on a LAF context, instead of a lattice. 

The abstract domain is based on term rewriting. It performs a dynamic
translation from a source term to a destination term, used as the
abstract state; the concretization of this state is simply its
collecting semantics. The domain is parametrized by a set of term
rewriting rules \(\mathcal{R}\).

\begin{figure}[htbp]

\[
\begin{array}{l@{\hspace{2mm}}l}
\varepsilon^\sharp = & [] \\
\llbracket [] \rrbracket^\sharp( C^\sharp) = & C^\sharp \\
\llbracket \textrm{let\ } x \triangleq op(x_1,\ldots,x_n) \textrm{\ in\ } C \rrbracket^\sharp(C^\sharp) = & \llbracket C \rrbracket^\sharp( C^\sharp[\textrm{let\ } x \triangleq \mathcal{R}({op}(x_1,\ldots,x_n))]) \\
\llbracket \textrm{let\ } x \triangleq \textrm{nondet}(x_1,x_2) \textrm{\ in\ } C \rrbracket^\sharp(C^\sharp) = & \llbracket C \rrbracket^\sharp( C^\sharp[\textrm{let\ } x \triangleq \mathcal{R}({\textrm{nondet}}(x_1,x_2))]) \\
\llbracket \textrm{let\ } x \triangleq \textrm{unknown}_{\mathcal{S}}() \textrm{\ in\ } C \rrbracket^\sharp(C^\sharp) = & \llbracket C \rrbracket^\sharp( C^\sharp[\textrm{let\ } x \triangleq \mathcal{R}({\textrm{unknown}_{\mathcal{S}}}())]) \\
\llbracket \textrm{let\ } x \triangleq \textrm{assume}(x_c,x_v) \textrm{\ in\ } C \rrbracket^\sharp(C^\sharp) = & \llbracket C \rrbracket^\sharp( C^\sharp[\textrm{let\ } x \triangleq \mathcal{R}({\textrm{assume}}(x_c,x_v))]) \\
\llbracket \textrm{let\ } x \triangleq \muterm{x_s}{C_b}{x_e}{x_i} \textrm{\ in\ } C \rrbracket^\sharp(C^\sharp) = 
& \textrm{let\ } C_b^\sharp = \llbracket C_b \rrbracket^\sharp([]) \textrm{\ in} \\
& \llbracket C \rrbracket^\sharp( C^\sharp[\textrm{let\ } x \triangleq  \mathcal{R}(\muterm{x_s}{C_b^\sharp}{x_e}{x_i})])


\end{array}
\]

\[ \Gamma \in \gamma(C^\sharp) \quad\Leftrightarrow\quad \Gamma \in \llbracket C^\sharp \rrbracket(\varepsilon) \]

\caption{Rewriting abstract domain parametrized by a term rewriting system $\mathcal{R}$}
\label{fig:term-rewriting-abstract-domain}

\end{figure}

We give term rewriting rules using pattern matching, on a
representation of terms where \(\textrm{let}\) bindings have been
inlined. Simple term rewriting rules include \(x \land x \to x\), \(1 * x
\to x\), \(2 + x + 3 \to x + 5\). However, as is, rules such as \(0 * x
\to 0\) or \(x - x \to 0\) are not correct for this domain. Indeed,
consider the terms:

\("\textrm{let\ } u \triangleq \textrm{unknown}_\mathbb{Z} \textrm{\ in\ let\ } v
\triangleq \textrm{assume}(u \ne 3, u) \textrm{\ in\ let } s \triangleq 0 * v \textrm{\
in s}"\).

and the result of applying the rule \(0 * x \to 0\) to it:

\("\textrm{let\ } u \triangleq \textrm{unknown}_\mathbb{Z} \textrm{\ in\ let\ } v
\triangleq \textrm{assume}(u \ne 3, u) \textrm{\ in\ let } s \triangleq 0 \textrm{\
in s}"\).

The possible environment \([u \mapsto 3, v \mapsto \bot, s \mapsto
\bot]\) for the first term has been replaced by the environment \([u
\mapsto 3, v \mapsto \bot, s \mapsto 0]\) in the second. One way to
deal with this issue is to propagate the "assume" conditions in the
replacement. Another is to accept these "over-approximating rewrites",
by changing the definition of \(\gamma\) to:
\begin{equation} \label{ref:eq-term-rewriting-over-approximating-gamma}
\Gamma \in \gamma(C^\sharp) \quad\Leftrightarrow\quad \exists \Gamma' \in \llbracket C^\sharp \rrbracket(\varepsilon): \forall x \in \Gamma: \Gamma[x] = \Gamma'[x] \lor \Gamma[x] = \bot 
\end{equation}
The latter allows very aggressive rewrites, such as \(x/x \to 1\) or \(1
/x < 2 \to \textrm{true}\), that disregard side conditions, usually a major difficulty of term-rewriting.

The soundness of both "exact" and "over-approximating" abstract
domains is established by considering only the term rewriting rules:

\todo{Déplacer une partie de tout ceci en annexe, en ayant un théorème simplifié: le domaine abstrait est correct si chacun des rewrite est correct}

\begin{definition} A term rewriting rule $l \to r$ is an \emph{exact} rewrite if, for any substitution $\sigma$ of the free variables of $l$, the evaluation of $l\sigma$ 
equals the evaluation of $r\sigma$. It is \emph{over-approximating} if they are equal when the evaluation of $l\sigma$ does not return $\bot$.
\end{definition}

\begin{theorem} If every term rewriting rule in $\mathcal{R}$ is \emph{exact}, then the abstract domain of Figure \ref{fig:term-rewriting-abstract-domain} is sound.
If every term rewriting rule in $\mathcal{R}$ is \emph{over-approximating}, then the abstract domain of Figure \ref{fig:term-rewriting-abstract-domain}, with $\gamma$ given by equation \ref{ref:eq-term-rewriting-over-approximating-gamma}, is sound.
\end{theorem}

\maybe{Permet de détecter des égalités à moindre coût.}

The term rewriting and non-relational domain represent very different
domains, one semantic in nature, the other more symbolic. We will now
see a more complex domains that combine and extend theses two models:
the constraint propagation abstract domain.

\section{A constraint propagation abstract domain}
\label{sec:orgc409c2f}
\label{sec:constraint-propag-abs-domain}

\maybe{ This domain is also central place of
exchange of information between abstract domains (as noted by
\citet{chang2005abstract}).}

\maybe{This domain is also central to exchanging
information between abstract domains; for instance numerical
invariants can be expressed as a mapping from expressions of 
expressible a}

The constraint domain is a good example of the advantages of our
approach: benefiting from the structure of LAF terms (including
targeted join and widening) that allows a single store implementation,
it can propagate and combine semantic information across the whole
program, in an efficient way. It is made of three elements:
\begin{itemize}
\item A LAF context \(C^\sharp\), called the \emph{constraints} (middle of Figure
\ref{fig:constraint-domain-example}). It can be seen as a rewrite of
the input LAF context with particular locations for \emph{assume}
definitions.
\item Two mappings \(\Gamma^\sharp_c\) and \(\Gamma^\sharp_v\) from variables
of the input context to variables of the constraints. This mapping
is made such that \(x\) and
\("\textrm{assume}(\Gamma^\sharp_c(x),\Gamma^\sharp_v(x))"\) evaluate
to the same values (the mapping and input term are not shown on
Figure \ref{fig:constraint-domain-example}).4
\item A mapping \(M\) from each variable of the constraints to a \emph{condition
map}, representing the possible values for the variable according to
some conditions (right of Figure
\ref{fig:constraint-domain-example}). How these conditions are
chosen depends on a strategy; ours is detailed below.
\end{itemize}

\smallskip
\noindent {\bf Generation of constraints\ } The construction of the constraints, \(\Gamma^\sharp_c\) and
\(\Gamma^\sharp_v\) is pretty straightforward (Figure
\ref{fig:constraints-generation}). The general idea is that the
constraints could represent the input term stripped from assume
expressions; the condition is instead stored in the \(\Gamma^\sharp_c\)
map.
 \maybe{these conditions are those used in the condition maps of
$M$.}
Actually we do not strip the assume statements entirely, but delay
them until the end of loops, or just before a "nondet", so that the
term used as the constraints does not lose information with regards
to the input term.

\floatstyle{plain}
\restylefloat{figure}
\begin{figure}[htbp]

\hspace{-9mm}\fbox{\begin{minipage}{1.15\textwidth}
\vspace{-3mm}
\[\begin{array}{l@{\hspace{2mm}}l}

\varepsilon^\sharp = & [] \hfill (x' \textrm{\ is\ fresh})\\
\llbracket [] \rrbracket^\sharp( C^\sharp,\Gamma_v^\sharp,\Gamma_c^\sharp ) = & \langle C^\sharp,\Gamma_v^\sharp,\Gamma_c^\sharp \rangle \\
\llbracket \textrm{let\ } x \triangleq op(x_1,\ldots,x_n) \textrm{\ in\ } C \rrbracket^\sharp( C^\sharp, \Gamma_v^\sharp, \Gamma_c^\sharp ) = 
& \llbracket C \rrbracket^\sharp( \begin{array}[t]{l} C^\sharp[\textrm{let\ } x' \triangleq {op}(\Gamma_v[x_1],\ldots,\Gamma_v[x_n])], \\
                                                     \Gamma_v^\sharp[x \mapsto x'], \Gamma_c^\sharp[x \mapsto \Gamma_c^\sharp[x_1] \land \ldots \land \Gamma_c^\sharp[x_n]]) \end{array} \\
\llbracket \textrm{let\ } x \triangleq \textrm{nondet}(x_1,x_2) \textrm{\ in\ } C \rrbracket^\sharp( C^\sharp, \Gamma_v^\sharp, \Gamma_c^\sharp ) = 
& \llbracket C \rrbracket^\sharp( \begin{array}[t]{l} C^\sharp[\textrm{let\ } x' \triangleq \textrm{nondet}(\begin{array}[t]{l} 
                                                        \textrm{assume}(\Gamma^\sharp_c[x_1],\Gamma^\sharp_v[x_1]),\\ 
                                                        \textrm{assume}(\Gamma^\sharp_c[x_2],\Gamma^\sharp_v[x_2]))], \end{array} \\
                                                     \Gamma_v^\sharp[x \mapsto x'], \Gamma_c^\sharp[x \mapsto \Gamma_c^\sharp[x_1] \lor \Gamma_c^\sharp[x_2]]) \end{array} \\
\llbracket \textrm{let\ } x \triangleq \textrm{unknown}_\mathcal{S}() \textrm{\ in\ } C \rrbracket^\sharp( C^\sharp, \Gamma_v^\sharp, \Gamma_c^\sharp ) = 
& \llbracket C \rrbracket^\sharp( \begin{array}[t]{l} C^\sharp[\textrm{let\ } x' \triangleq \textrm{unknown}_{\mathcal{S}}()], \\
                                                     \Gamma_v^\sharp[x \mapsto x'], \Gamma_c^\sharp[x \mapsto \textrm{true}]) \end{array} \\
\llbracket \textrm{let\ } x \triangleq \textrm{assume}(x_c,x_v) \textrm{\ in\ } C \rrbracket^\sharp( C^\sharp, \Gamma_v^\sharp, \Gamma_c^\sharp ) = 
& \llbracket C \rrbracket^\sharp( \begin{array}[t]{l} C^\sharp, \Gamma_v^\sharp[x \mapsto \Gamma^\sharp_v[x_v]],\\
                                                      \Gamma_c^\sharp[x \mapsto \Gamma_c[x_c] \land \Gamma_c[x_v] \land \Gamma_v[x_c]]) \end{array} \\
\llbracket \textrm{let\ } x \triangleq \muterm{x_s}{C_b}{x_e}{x_i} \textrm{\ in\ } C \rrbracket^\sharp( C^\sharp, \Gamma_v^\sharp, \Gamma_c^\sharp ) = 
& \textrm{let\ } \langle C_b^\sharp,\Gamma_v^{\prime\sharp},\Gamma_c^{\prime\sharp} \rangle = \llbracket C_b \rrbracket^\sharp([], \Gamma_v^\sharp,\Gamma_c^\sharp) \textrm{\ in} \\
& \textrm{let\ } E = "{\textrm{assume}(\Gamma_c^{\prime\sharp}[x_e],\Gamma_v^{\prime\sharp}[x_e])}" \textrm{\ in} \\
& \llbracket C \rrbracket^\sharp( \begin{array}[t]{l} C^\sharp[\textrm{let\ } x' \triangleq \muterm{x_s}{C_b^\sharp}{E}{\Gamma_v^\sharp[x_i]})], \\
                                                     \Gamma_v^\sharp[x \mapsto x'], \Gamma_c^\sharp[x \mapsto \Gamma_c^\sharp[x_i]]) \end{array} \\

\end{array}
\]%
\end{minipage}}
\vspace{-3mm}
\caption{Generation of the constraints $C^\sharp$ and the maps $\Gamma_c^\sharp$ and $\Gamma_v^\sharp$.}
\label{fig:constraints-generation}
\end{figure}
\floatstyle{boxed}
\restylefloat{figure}

\smallskip
\noindent {\bf Initial evaluation\ } 
We now describe the construction of the mapping \(M\) of condition
maps. An \emph{initial evaluation} of an input variable \(x\) is done
concurrently with constraints generation: the possible values for
\(\Gamma^\sharp_v[x]\) is computed using its definition, for the
condition \(\Gamma^\sharp_c[x]\). For instance in Figure
\ref{fig:constraint-domain-example}, the initial evaluation of \(xdiv\)
is done with condition \(c_2\); we first retrieve the abstract value for
\(x\) with condition \(c_2\) (which is \([-8;-1] \dot\sqcup [0;8] =
[-8;8]\)); then the value for \(9\) (which is \([9;9]\)); then create the
binding \(xdiv \mapsto c_2 \Vdash [0;0]\), where \([0;0] =
[-8;8]\dot/[9;9]\).

\smallskip
\noindent {\bf Constraint propagation\ }
When an \(\textrm{assume}(c,x)\) definition is evaluated in the input
term, we perform a \emph{constraint propagation}. It consists in
reevaluating definitions such as \("\textrm{let\ } c \triangleq x < 0"\),
refining and using the information attached to the variables
corresponding to the result and arguments of the operator. The
algorithm is similar to constraint satisfaction algorithms such as
AC-3 \cite{mackworth1977consistency}; but we maintain, together with
the worklist of variables whose abstract value has changed, the
condition for which they have changed. The constraint propagation of
condition \(c\) is initiated by adding the binding \(c \mapsto c \Vdash
\{true \}\).  In Figure \ref{fig:constraint-domain-example}, a first
chain of constraints propagation is the addition of the bindings \(c_1
\mapsto c_1 \Vdash \{true \}\), then \(x \mapsto c_1 \Vdash
[{-}\!\infty;{-}1]\).

When an \(\textrm{assume}(c',\ldots)\) constraint is traversed, the \(c'\)
condition is added as a conjunct to the condition being
propagated. This is seen in the addition of \(c_1\) in the constraint
propagation chain \(c_2 \mapsto c_2 \Vdash \{true \}\), \(abs \mapsto c_2
\Vdash [0;8]\), \(nx \mapsto c_1 \land c_2 \Vdash [1;8]\), \(x \mapsto c_1
\land c_2 \Vdash [{-}8;{-}1]\).

The constraint propagation phase terminates if the lattices used in
the abstract value cannot be indefinitely refined (e.g. refining \(y
\le y/2\) using an interval of rational numbers). But it is always
sound to limit the number of propagation, and we evaluate different
heuristics in Section \ref{sec:experimental-evaluation}.

\smallskip
\noindent {\bf Loops\ }
are handled like in the non-relational analysis of Section
\ref{sec:non-rel-abs-int}; the condition map, seen as a function
lattice from conditions to abstract values, is used to join, widen and
test for inclusion the loop input and output. The only difference is
that the conditions defined inside the loop body do not have any
meaning in the next iteration, or outside of the loop; thus as soon as
the loop body has been fully evaluated, these conditions are
removed by existential quantification.

\smallskip
\noindent {\bf Concretization\ } 
The concretization is best defined as the composition of two
parts. The first relates the input term to the generated constraints,
and is similar to the one of rewriting-based abstract domains of
Section \ref{sec:simple-rewriting-abs-int}. 
\maybe{Note that, in the
implementation described above, the input term and the constraints are
actually equivalent: this "pre-processing step" does not lose any
information.}
\begin{multline*}
\Gamma \in \gamma_1(\langle C^\sharp,\Gamma^\sharp_c,\Gamma^\sharp_v \rangle) \ \Leftrightarrow\\[-6pt] \exists \Gamma' \in \llbracket C^\sharp \rrbracket(\varepsilon): \forall x \mapsto v \in \Gamma:
\left\{ \begin{array}{l@{\hspace{4mm}}r} \Gamma'[\Gamma^\sharp_c[x]] = \textrm{true} \ \land\ \Gamma'[\Gamma^\sharp_v[x]] = v & \textrm{if\ } v \ne \bot \\ 
\Gamma'[\Gamma^\sharp_c[x]] = \textrm{false} & \textrm{if\ } v = \bot
\end{array} \right.
\end{multline*}

The second part relates the term in the constraints to the values
contained in the map \(M\). This amounts to seeing constraint generation
as a mere pre-processing of the input. It is similar to that of the
non-relational abstract domain of \ref{sec:non-rel-abs-int}, but
taking conditions into account.
\[ \Gamma \in \gamma_2(M) \quad\Leftrightarrow\quad \forall x \in \Gamma, c \Vdash v^\sharp \in M[x]:\  \Gamma[c] = \textrm{true} \ \Rightarrow\ \Gamma[x] \in \gamma_{\mathcal{V}}(v^\sharp) \]

The combination consists in replacing, in the definition of
\(\gamma_1\), the set of possible environments \(\llbracket C^\sharp \rrbracket(\varepsilon)\), by the approximation of
this set \(\gamma_2(M)\).

\begin{theorem} The constraint propagation abstract domain is sound. \end{theorem}

\maybe{ - Deux versions de la propagation avant et arrière:
   - Version path-insensitive: pour évaluer en avant, pour une condition, on joine tout ceux qui correspondent. Puis on applique l'opération (avec un point au dessus). De même pour l'arrière.
   - Version path-sensitive: on évalue 2 par 2 toutes les combinaisons qui matchent. }

\maybe{Note: pour interpretation abstraite, on génére le assume au début,
comme ca on travaille avec les domaines réduits aussi vite que
possible.}

\section{An abstract interpreter of embedded C programs}
\label{sec:org03328a1}
\label{sec:complete-example-abs-inter}

This section demonstrates the practical applicability of our approach,
by describing the implementation of a complete analyzer for embedded C
programs (including low-level memory manipulation such as casts and
bitfields, but currently excluding recursion and dynamic memory
allocation). The system is composed as a succession of simple abstract
domain "passes" communicating with one another. 

\begin{figure}[htbp]
\centering
\pgfdeclarelayer{background}
\pgfdeclarelayer{foreground}
\pgfsetlayers{background,foreground}
\scalebox{1.0}{\begin{tikzpicture}[xscale=3,yscale=2]

\coordinate (C) at (-1,0) ;
\begin{pgfonlayer}{foreground}
\node[draw,ellipse,rounded corners, align=center] (c2laf) at (0,0) {\small Static \\ translation};
\node[fill=black!20,draw,rectangle,rounded corners, align=center] (region) at (1,0) {\small Region\\separation};
\node[fill=black!20,draw,rectangle,rounded corners, align=center] (array) at (2,0) {\small Array\\separation};
\end{pgfonlayer}

\begin{scope}[shift={(0,-0.2)}]
\begin{pgfonlayer}{foreground}
\node[fill=black!20,draw,rectangle,rounded corners, align=center] (bitvec) at (2,-1) {\small Bitvector\\ simplification};
\node[fill=black!20,left,draw,rectangle,rounded corners, align=center] (constr) at (0.9,-0.75) {\small Constraint\\ propagation};
\node[fill=black!20,left,rectangle,rounded corners, align=center] (smt) at (0.9,-1.25) {\small Generation of \\proof obligation};
\node[fill=white,draw,ellipse,rounded corners, align=center] (smtsolv) at (-0.6,-1.25) {\small SMT solver};
\node[fill=black!20,rectangle,rounded corners=2pt,draw,minimum size = 1pt,inner sep=1.5pt](middle) at (1.2,-1) {$\times$};
\end{pgfonlayer}
\begin{pgfonlayer}{background}
\draw[rounded corners,fill=black!20] (smt.north east) rectangle (smt.south -| smtsolv.west);
\end{pgfonlayer}
\end{scope}


\begin{pgfonlayer}{foreground}
\draw[->] (C) -- node[above,align=center] {C program} (c2laf);
\draw[->] (c2laf) edge node[below,align=center] {LAF\\ (memory + \\ bitvector)} (region);
\draw[->] (region) edge node[below,align=center] {LAF \\ (array + \\ bitvector)} (array);
\draw[->] (array) edge node[right,align=center] {LAF\\ (bitvector)} (bitvec);
\draw[->] (bitvec) edge node[below=5pt,align=center] {LAF\\ (bitvector)} (middle);
\draw[->] (middle) |- (constr);
\draw[->] (middle) |- (smt);
\draw[<->] (smtsolv.east) +(0.05,0) -- node[above] {\scriptsize SMT} (smt);
\end{pgfonlayer}

\end{tikzpicture}}

\caption{High-level view of the analyzer. Gray rectangles represent abstract domains, and ellipses other processes. $\times$ represents abstract domain product: both domains have the same input term.  }
\label{fig:high-level-view-analyzer}

\end{figure}
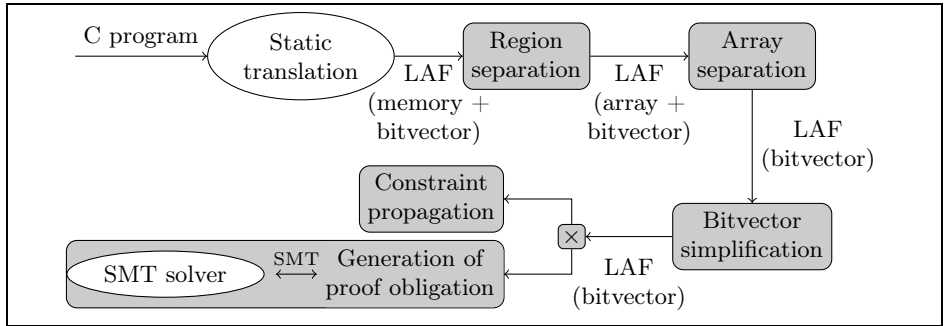

Figure \ref{fig:high-level-view-analyzer} presents a high-level view
of the analyzer. Each rectangular node is an abstract domain, that
inputs a LAF term, and possibly outputs a simplified one. Input and
output terms may use operators of different theories. Especially, the
LAF term obtained from the translation of the C program contains
\texttt{load} and \texttt{store} memory operators; while the LAF term used as input
of the leaves abstract domains do not. This allows 1. to translate
this term to SMT solvers that only understand bitvector theory,
and 2. to implement a numeric constraint propagation domain that is
unaware of memory operations.

The translation from C to LAF is standard (similar to
e.g. \citet{cytron1991efficiently}). Appendix
\ref{sec:translation-from-while} presents the full translation rules
for the simpler \textsc{While} language.

\smallskip
\noindent {\bf Region and cell separation\ } The region separation functor splits the memory into independent,
non-overlapping memory regions, where a region corresponds to the
memory allocated for a C variable (local or global). The cell
separation functor partitions a memory region into contiguous slices
of known size, and was informally described in Section
\ref{sec:hierarchy-translator-domains}. Both are currently quite
naive; but better memory representations functors
(e.g. \cite{chang2013modular,cousot2011parametric}) could be adapted
to become translator domains.

The output of these successive translator domains is a LAF term that
does not contain memory operations, and whose variables correspond to
the values contained in memory. This term is analyzed numerically by
the later domains; the result of this analysis is itself used by the
memory translator domains to know which memory locations are being
read and written.

\maybe{Interet de splitter ces domaines: plus simple, et on peut rajouter la recency abstraction pour l'un, des array summarizations pour l'autre;}
\maybe{}

\smallskip \noindent {\bf Purely syntactic term rewriting\ }
Our prime use of term rewriting is to simplify the bitvector
concatenation and extraction operations. This happens notably when the
program performs byte per byte copies (e.g. using \texttt{memcpy}). It plays
the same role as the memory equality predicate domain of
\citet{mine2006field}, but the implementation is very different, as
equality between values is provided by syntactic equality between LAF
variables.

\smallskip
\noindent {\bf Constraint propagation\ } was already described in section \ref{sec:constraint-propag-abs-domain}

\notnow{In our current implementation, we first convert bitvector operations to arithmetical ones, and perform constraint propagation using an integer theory. For this we currently assume (and check) that bitvector operations do not overflow, which is sometimes limiting; a future version will  (assuming there is no overflow).}

\label{sec:mutual-refiniement-SMT}

\smallskip \noindent {\bf Generation of proof obligations as
first-order and Horn formulas\ } translate LAF terms into a set of
clauses, such that satisfiability testing allows to test if a variable
can have a specific value. Appendix \ref{sec:translation-to-smt}
details the formal translation of LAF term into 1st-order and Horn
clauses; the translation to Horn clauses is exact, while the
translation to first-order clauses is an over-approximation.

\todo{ but the latter yields a simpler, decidable problem (on bitvectors).}

Note that this results in a true combination of abstract domains, and
not just a first static analysis pass followed by generation of proof
obligations. The benefit is that an abstract domain can be combined
with others so that the combination is mutually beneficial. In
particular:

\begin{enumerate}
\item The generation of proof obligation domain benefits from the
simplification of translator domains. In the analyzer, the
generated SMT formula only refers to bitvector operations, and never
to memory operations; this is beneficial since SMT solvers are much
better at handling integers and bitvectors than at handling memory
operations.

\item The domain also indirectly benefits from the numeric constraint
propagation domain, as it is used by the translator domains to
perform simplifications. It could also use the invariants computed
by that domain to increase precision or speed up resolution.

\item In turn, the domain can be used to simplify complex boolean
expressions into the \(\textit{true}\) or \(\textit{false}\) constants,
allowing for a precision increase in the other domains. It can
also be used e.g. to refine the bounds of intervals.
\end{enumerate}

\section{Experimental evaluation}
\label{sec:orge917583}
\label{sec:experimental-evaluation}

This section presents preliminary benchmarks of our analyzer. Please
keep in mind that it is still a prototype; it lacks some of
the features of C (dynamic memory allocation, recursion, function
pointers and \texttt{longjmp}); the handling of function calls (by recursive
inlining) is still naive; no optimization of hot code path was performed.

\begin{table}[htbp]
\footnotesize
\hspace{-10mm}\begin{tabular}{l|r|r|r||r|r||r|r|r||r|r|r||r|r|r}
\multicolumn{4}{r||}{Propagation limit}     & \multicolumn{2}{c||}{0} & \multicolumn{3}{c||}{1} & \multicolumn{3}{c||}{2} & \multicolumn{3}{c}{$\infty$ (backward)} \\ \hline
Bench & \#LOC & \#expr & \#A & time & UP & time & UP & improv. & time & UP & improv. & time & UP & improv.\\
\hline
adpcm        &    610& 499 & 1249 & 1.12 & 39 & 0.94 & 39 & 56\% & 1.04 & 22 & 18\% & 1.06 & 22 & 0.0\%\\
2048         &    435& 617 & 8778 & 38.5 & 314 & 32.9 & 314 & 19\% & 33.19 & 100 & 94\% & 32.56 & 100 & 0.6\%\\
papabench0.4 &   4983& 8286 & 6415 & 84.4 & 412 & 78.2 & 412 & 0.9\% & 82.14 & 83 & 30\% & 80.96 & 83 & 0.0\%\\
transport    &  13294& 10674 & 11064 & 213 & 100 & 182 & 100 & 8.3\% & 61.41 & 14 & 16\% & 63.04 & 14 & 0.1\%\\
robotics \#0 &  13353& 161 & 306 & 1.95 & 18 & 1.74 & 18 & 26\% & 1.79 & 18 & 0.0\% & 1.80 & 18 & 3.7\%\\
robotics \#1 &  13353& 1454 & 1978 & 4.06 & 80 & 3.61 & 80 & 1.9\% & 2.85 & 51 & 14\% & 2.77 & 51 & 0.7\%\\
robotics \#2 &  13353& 22 & 125 & 0.88 & 0 & 0.78 & 0 & 0.0\% & 0.84 & 0 & 0.0\% & 0.85 & 0 & 0.0\%\\
robotics \#3 &  13353& 872 & 860 & 0.97 & 0 & 0.85 & 0 & 1.1\% & 0.94 & 0 & 0.0\% & 0.94 & 0 & 0.0\%\\
robotics \#4 &  13353& 298 & 620 & 3.24 & 106 & 2.76 & 106 & 8.1\% & 1.46 & 62 & 76\% & 1.37 & 21 & 22\%\\
robotics \#5 &  13353& 2372 & 375 & 98.2 & 221 & 85.9 & 221 & 3.2\% & 77.68 & 167 & 27\% & 74.70 & 165 & 0.9\%\\
robotics \#6 &  13353& 280 & 481 & 0.96 & 6 & 0.85 & 6 & 0.0\% & 0.94 & 6 & 0.0\% & 0.92 & 6 & 0.0\%\\
\hline
\end{tabular}
\#LOC = lines of code; \#expr = live expressions; \\ \#A = total alarms to be proved; UP = alarms left unproved; \\ improv. = \% of expressions more precise than the preceding propagation limit.
\vspace{2mm}

\caption{Benefits of constraint propagation.}
\label{tab:constraint-propagation}
\end{table}

\noindent {\bf Analyzing embedded industrial applications\ }
Table \ref{tab:constraint-propagation} demonstrates the benefits of
using constraint propagation on a variety of real benchmarks. The
first columns give the number of expressions of the program that are
not dead, and the number of memory-related alarms that must be proved
(i.e. pointers are valid and indices are in bounds). Then, for
different limits on the constraint propagation, we give the analysis
time (in seconds); the number of alarms that remain unproved; and the
number (and percentage) of expressions for which we compute a more
precise set of possible values, wrt. the preceding propagation
limit. The first three benchmarks are open-source (2048 is a game;
adpcm a filter; papabench0.4 represents code embedded in typical
UAVs). All the other benchmarks are automatic control systems coming
from various industries (the last is made of 7 independent threads).
\looseness=-1

The benchmarks demonstrate that constraint propagation has a real
impact on the precision of the analysis. Going from 0 to 2 variables
being propagated leads to refining most of the expressions (and
alarms) of the programs, while \emph{decreasing} the analysis time. This
can be explained by the fact that less dead code is visited, but also
in our case that memory updates are more precise and concern less
locations. Going for unlimited backward propagation (or backward +
forward, not shown in the table) also keep on being faster and more
precise, but to a lesser degree. Note that the number of alarms is
still high, but can be reduced to almost 0 by standard tricks such as
loop unrolling and user annotations.

\smallskip

\noindent {\bf Generation of Horn clauses\ } Another way to
discharge unproved assertions is to send the remaining ones to a
solver which focuses on a single goal. We made an experiment with
benchmarks of the SVComp 2016 competition \cite{beyer2016reliable},
whose goal is to prove the validity of an assertion in a program which
range from few to tens of KLOCS. Most of these assertions are out of
the reach of our abstract domain using intervals. However, the
structure of out interpreter allows to build a LAF term, which is
equivalent to the original program, but stripped from any memory
operations. This property allows to leverage (after conversion, see
Appendix \ref{sec:translation-to-smt}) the Z3 horn checker \(\mu{}Z\)
\cite{hoder2012generalized}, which does not support array theory, only
bitvectors or integers.

\begin{table}
\begin{center}
\begin{tabular}{l|rrr|rrr|rrr|r}
               
& \multicolumn{3}{c|}{Lines of code} & \multicolumn{3}{c|}{Buggy programs} & \multicolumn{3}{c|}{Correct programs} & \\
SVComp category & min & avg & max &total & proved & unsound & total & proved & unsound & points \\
\hline
Product lines & 838 & 1943 & 3789 & 265 & 92 & 0 & 332 & 262 & 0 & 616 \\
Loops         & 14  & 46   & 1644 & 48 & 18 & 0 & 93 & 42 & 0 & 102 \\
ControlFlow   & 94  & 634  & 2152& 18 & 3 & 0 & 30 & 15 & 0 & 33 \\
\end{tabular}
\end{center}

\caption{Number of programs proved buggy or correct, in some SVComp categories, with a 10s timeout. We count 2 points 
for every correct program proved correct, and 1 for buggy program proved buggy. Our tool never provides a wrong (unsound) answer.}
\end{table}
On the two first categories and with a 10s timeout, the abstract
interpreter combined with \(\mu{}Z\) is already
competitive\footnote{\url{https://sv-comp.sosy-lab.org/2016/results/results-verified/}}
(would rank 3rd/18 on Loops, and 3rd/16 on ControlFlow). This is
despite all the shortcomings of our current implementation: some C
features are not supported (e.g. variable length arrays), and we do
not yet export the invariants we found to help in the Horn clauses.

\section{Related work}
\label{sec:orge0b10d9}

The closest relatives to our term-based abstract interpretation
framework are existing symbolic abstract domains in the traditional
lattice-based framework. 

The symbolic constant domain of \citet{mine2007symbolic} maps program
variables to expressions on other program variables, and provides a
solution to the loss of precision induced by storing intermediate
computations in temporary variables. \citet{logozzo2008relative} and
\citet{djoudi2016recovering} implement similar domains, and insists on
the need of these methods for low-level languages, in which every
computation makes use of temporary variables.

\citet{chang2005abstract} introduces a symbolic abstract domain where
variables represent values, instead of memory locations. This avoids
loosing precision when program variables are overwritten. Numeric
abstract domains are used to compute relations between these symbolic
variables. \citet{chang2013modular} shows the importance of having
variables representing values when designing precise memory
abstractions. In term-based abstract interpretation, base abstract
domains compute relations between LAF variables (which represent
values); these variables are referred to by our memory
abstraction. The main difference is that all the LAF variables are
linked together in a term which is an abstraction of the whole
program, i.e. we never lose any symbolic information, even on loops
and control-flow joins.

\citet{gange2016uninterpreted} combines their symbolic abstract domain
with constraint propagation over non-relational domains. As LAF terms
represent loops and conditionals precisely, our constraint
propagation abstract domain extends this work with the ability to
propagate the constraints across the whole program.

All these domains belong to the traditional lattice-based framework;
thus none of them leverage the fact that variables represent values by
sharing all the information about them in a single store.

\section{Conclusion}
\label{sec:org067848f}

We have presented a term-based abstract interpretation framework,
whose main ingredients are: a logic, that can be used as the abstract
state in an abstract domain, and can represent the relations between
values in the program without loss of precision; and the definition of
abstract domains as abstract interpreters over this logic, allowing
the definition of abstract domains as a combination of translations.
We have demonstrated the applicability of the framework by describing
several abstract domains combining numeric and symbolic reasoning; and
we used these domains to build a complete analyzer for C programs. We
now plan on applying the technique on languages where symbolic
reasoning is very important, such as Static Single Assignment or
binary analysis.

\bibliography{biblio}

\appendix

\section{Small-step operational semantics of LAF}
\label{sec:org6b2b754}
\label{sec:laf-small-step-operational-semantics}

A small-step operational semantics also exists for LAF terms:

\begin{align*}
\Sigma::\langle \Gamma, \textrm{let\ } x = op_n(x_1,\ldots,x_n) \textrm{\ in\ } t \rangle     \to & \\ & \hspace{-5cm} \Sigma::\langle \Gamma[x \mapsto op_n(\Gamma[x_1],\ldots,\Gamma[x_n])], t \rangle\\
\\
\Sigma::\langle \Gamma, \textrm{let\ } x = \textrm{nondet}(x_1,x_2) \textrm{\ in\ } t \rangle \to & \\ & \hspace{-5cm} \Sigma::\langle \Gamma[x \mapsto \Gamma[x_1]], t \rangle\\
\Sigma::\langle \Gamma, \textrm{let\ } x = \textrm{nondet}(x_1,x_2) \textrm{\ in\ } t \rangle \to & \\ & \hspace{-5cm} \Sigma::\langle \Gamma[x \mapsto \Gamma[x_2]], t \rangle\\
\\
\Sigma::\langle \Gamma, \textrm{let\ } x = \textrm{assume}(x_c,x_v) \textrm{\ in\ } t \rangle \to & \\ & \hspace{-5cm} \Sigma::\langle \Gamma[x \mapsto \Gamma[x_v]], t \rangle \mathrm{\ when\ } \Gamma[x_c] = true \\
\Sigma::\langle \Gamma, \textrm{let\ } x = \textrm{assume}(x_c,x_v) \textrm{\ in\ } t \rangle \to & \\ & \hspace{-5cm} \Sigma::\langle \Gamma[x \mapsto \bot], t \rangle \mathrm{\ when\ } \Gamma[x_c] = false \\
\\
\Sigma::\langle \Gamma, \textrm{let\ } x = \muterm{x_s}{t_b}{x_e}{x_i} \textrm{\ in\ } t \rangle     \to & \tag{do not enter loop} \\
& \hspace{-5cm} \Sigma::\langle \Gamma[x \mapsto \Gamma[x_i]], t \rangle  \\
\Sigma::\langle \Gamma, \textrm{let\ } x = \muterm{x_s}{t_b}{x_e}{x_i} \textrm{\ in\ } t \rangle     \to & \tag{enter loop} \\
& \hspace{-5cm} \Sigma::\langle \Gamma, \textrm{let\ } x = \muterm{x_s}{t_b}{x_e}{x_i} \textrm{\ in\ } t \rangle:: \langle \Gamma[x_s \mapsto \Gamma[x_i]], t_b[x_e] \rangle  \\
\Sigma::\langle \Gamma, \textrm{let\ } x = \muterm{x_s}{t_b}{x_e}{x_i} \textrm{\ in\ } t \rangle:: \langle \Gamma', x_e \rangle  \to & \tag{loop exit} \\
& \hspace{-5cm} \Sigma::\langle \Gamma[x \mapsto \Gamma'[x_e]],t \rangle  \\
\Sigma::\langle \Gamma, \textrm{let\ } x = \muterm{x_s}{t_b}{x_e}{x_i} \textrm{\ in\ } t \rangle:: \langle \Gamma', x_e \rangle  \to & \tag{loop again} \\ 
& \hspace{-5cm} \Sigma::\langle \Gamma, \textrm{let\ } x = \muterm{x_s}{t_b}{x_e}{x_i} \textrm{\ in\ } t \rangle::\langle \Gamma[x_s \mapsto \Gamma'[x_e]],t_b[x_e] \rangle  \\
\end{align*}

The semantics is nondeterministic because of the constructors \emph{nondet}
  and \(\mu\). The semantics uses a stack, whose depth represents the
  level of loop nesting in which we are. The stack is used to save the
  context when entering a loop.  Elements of the stack are pairs of an
  environment and a term, respectively representing the values already
  computed, and the term that remains to be computed, for each loop
  nesting level. Most step just update the head of the stack, by
  updating the environment and the term.

Execution "blocks" when an \texttt{assume} expression is encountered, with
a value of \emph{false} for its first argument. This is represented by
adding to all sorts a special \(\bot\) symbol. Note that \texttt{assume}
delimits the part of the term which is blocked; execution of other
subterms continue without any problem.

\begin{theorem}[Alternative definition of the collecting semantics] We note by $\to^{*}$ the transitive closure of $\to$. Then

\[ \llbracket C \rrbracket(\varepsilon) = \{\ \Gamma: \langle \varepsilon, C \rangle\to^{*}\langle \Gamma, [] \rangle\ \} \]
\end{theorem}

\section{Lifting traditional relational abstract domains}
\label{sec:org094f36f}
\label{sec:lifting-traditional-relational-abstract-domains}

When generating constraints (Figure \ref{fig:constraints-generation}),
for each input variable \(x\) we extracted a condition
\(\Gamma_c^\sharp[x]\), corresponding to the necessary condition for \(x\)
to evaluate to a value different from \(\bot\). The idea here is
similar: we associate to each input variable, an element \(D\) of a
traditional abstract domain; this element corresponds to environments
that match the condition \(\Gamma_c^\sharp[x]\), and describes the
relations between \(x\) and all the variables on which it depends
transitively (where \(a\) depends on \(b\) means that \(b\) is an argument
of the operator used to define \(a\)).

%
%
%
%
%
%
%
%
%
%
\begin{figure}[htbp]

\[
\begin{array}{l@{\hspace{2mm}}l}
\varepsilon^\sharp = & \varepsilon \\ \\
\llbracket \textrm{let\ } x = op(x_1,\ldots,x_n) \textrm{\ in\ } C \rrbracket^\sharp(\Gamma^\sharp) = 
& \textrm{let\ } D = \bigsqcap_{1 \le i \le n}\Gamma^\sharp[x_i] \textrm{\ in} \\
& \textrm{let\ } D' = \{| x \gets op(x_1,\ldots,x_n) |\}(D) \textrm{\ in} \\
& \llbracket C \rrbracket^\sharp(\Gamma^\sharp[x \mapsto D']) \\ \\ 
\llbracket \textrm{let\ } x = \textrm{nondet}(x_1,x_2) \textrm{\ in\ } C \rrbracket^\sharp(\Gamma^\sharp) = 
& \textrm{let\ } D_1 = \{| x \gets x_1 |\}(\Gamma^\sharp[x_1]) \textrm{\ in}\\
& \textrm{let\ } D_2 = \{| x \gets x_2 |\}(\Gamma^\sharp[x_2]) \textrm{\ in}\\
& \textrm{let\ } D' = D' = D_1 \sqcup D_2 \textrm{\ in} \\
& \llbracket C \rrbracket^\sharp(\Gamma^\sharp[x \mapsto D']) \\ \\
\llbracket \textrm{let\ } x = \textrm{unknown}() \textrm{\ in\ } C \rrbracket^\sharp(\Gamma^\sharp) = & \llbracket C \rrbracket^\sharp(\Gamma^\sharp[x \mapsto \top]) \\ \\

\llbracket \textrm{let\ } x = \textrm{assume}(x_c,x_v) \textrm{\ in\ } C \rrbracket^\sharp(\Gamma^\sharp) = 
& \textrm{let\ }D = \Gamma^\sharp[x_c] \sqcap \Gamma^\sharp[x_v] \textrm{\ in}\\
& \textrm{let\ } D' = \{| \textrm{assume\ } x_c |\}(D) \textrm{\ in}\\
& \textrm{let\ } D'' = \{| x \gets x_v |\}(D') \textrm{\ in}\\
& \llbracket C \rrbracket^\sharp(\Gamma^\sharp[x \mapsto D'']) \\ \\
 \texttt{killall}([],D) = & D \\
 \texttt{killall}(\textrm{let\ } x = ... \textrm{\ in\ } C,D) = & \texttt{killall}(C,\{| x \gets \textrm{unknown}() |\}(D)) \\ \\
\llbracket \textrm{let\ } x = \muterm{x_s}{C_b}{x_e}{x_i} \textrm{\ in\ } C \rrbracket^\sharp(\Gamma^\sharp) =
& \textrm{let\ } D_i = \{| x \gets x_i |\}(\Gamma^\sharp[x_i]) \textrm{\ in} \\
& \textrm{let\ rec\ \tt fix}(D) = \\
& \quad \textrm{let\ } D_s = \{ x_s \gets x \}(D) \textrm{\ in} \\
& \quad \textrm{let\ } \Gamma^\sharp_s = \Gamma^\sharp[x_s \mapsto D_s] \textrm{\ in} \\
& \quad \textrm{let\ } \Gamma^\sharp_b = \llbracket C_b \rrbracket^\sharp(\Gamma^\sharp_s) \textrm{\ in} \\
& \quad \textrm{let\ } D_e = \{| x \gets x_e |\}(\Gamma^\sharp_b[x_e]) \textrm{\ in} \\
& \quad \textrm{let\ } D_s' = D_i \sqcup \texttt{killall}(D_e) \textrm{\ in} \\
& \quad \textrm{if\ } (D_s' \sqsubseteq D_s) \textrm{\ then\ } D_s' \textrm{\ else\ } \texttt{fix}(D_s \nabla D_s') \\
& \textrm{in\ } \\
& \textrm{let\ } D' = \texttt{fix}(D_i) \textrm{\ in} \\
& \llbracket C \rrbracket^\sharp(\Gamma^\sharp[x \mapsto D'])
\end{array}
\]

\caption{Algorithm: Evaluating LAF terms with usual abstract domains. \mbox{$\{| x \gets expr |\}$} denotes the transfer function for the assignment $x \gets e$. }
\label{fig:conversion-from-usual-domains}

\end{figure}

We note by \(\gamma_t\) the concretization of the traditional
domain. The concretization is defined as follows: for every binding \(x
\mapsto v\) of possible environments \(\Gamma\), every abstract domain
element in \(\Gamma^\sharp\) must agree that this binding is indeed
possible.

\[ \Gamma \in \gamma(\Gamma^\sharp) \quad\Leftrightarrow\quad \forall (x \mapsto v) \in \Gamma: \forall (y \mapsto D) \in \Gamma^\sharp: \exists \Gamma' \in \gamma_t(D): \Gamma'[x] = v \]

\begin{theorem} If the original abstract domain is sound with regards to its concretization $\gamma_t$, then its lifting is also sound. 
\end{theorem}

Because this lifting relates all the variables of a LAF term; LAF
terms can contain a large number of variables; most operations on
numerical abstract domains have a complexity supra-linear in the
number of variables, a naive application of this technique would
probably be very slow. However, this can be mitigated by exploiting
the fact that most variables would be unrelated
\cite{gange2016exploiting}, or limiting relations by \emph{packing}
variables together \cite{blanchet2003static,venet2004precise}.

\section{Translation to SMT and Horn}
\label{sec:org10a762b}
\label{sec:translation-to-smt}

\newcommand{\smtmap}{\mathcal{M}}

\subsection{Translation to SMT}
\label{sec:orgcad4f08}

We begin by this translation, as it is easier

Intuitively, the translation associates to each LAF variable \(x\) a
pair of SMT variables \(\langle c_x,v_x \rangle\), where
\begin{itemize}
\item \(c\) represents the necessary condition for \(x\) to be different from \(\bot\)
\item \(v\) represents the value to which \(x\) evaluates.
\end{itemize}

Thus if \(x\) has value \(\bot\), then \(c_x\) has value \(false\) and \(v_x\)
can be anything; if \(x\) has value \(33\), then \(c_x\) is \(true\) and \(v_x\)
is also 33.

More formally, the translation \(\llparenthesis \cdot \rrparenthesis\)
creates a first-order formula \(\varphi\) and a mapping \(\smtmap\) from
LAF variables \(x\) to SMT variables \(\langle c_x,v_x\rangle\); such that
if \(x \Downarrow u\), then the formula \(\varphi \land c_x \land v_x =
u\) is satisfiable. The converse is not true, because the translation
of loops is over-approximated, as done in weakest-precondition
computation. This is "fixed" by generating Horn clauses instead of
SMT, which extends this translation to handle loops.

\begin{mathpar}
\inferrule[Empty and Result]{}{\llparenthesis [] \rrparenthesis(\smtmap,\varphi) = \langle \smtmap, \varphi \rangle}

\inferrule[Theory op]{\smtmap[x_1] = \langle c_1,v_1 \rangle \\ \ldots \\ \smtmap[x_n] = \langle c_n,v_n \rangle \\ c \textrm{\ fresh} \\ v \textrm{\ fresh}}
{\llparenthesis \mathrm{let\ } x = op_n(x_1,\ldots,x_n) \mathrm{\ in\ } C \rrparenthesis(\smtmap,\varphi) =\\ \llparenthesis C \rrparenthesis( \smtmap[x \mapsto \langle c, v\rangle],\  \varphi \land (c = c_1 \land \ldots \land c_n) \land (v = op_n(v_1,\ldots,v_n)))}\\


\inferrule[Unknown]{v \textrm{\ fresh}}
{\llparenthesis \mathrm{let\ } x = \textrm{unknown}_\mathcal{S}() \mathrm{\ in\ } C \rrparenthesis(\smtmap,\varphi) = \llparenthesis C \rrparenthesis(\smtmap[x \mapsto \langle true, v\rangle],\varphi)}\\

\inferrule[Assume]{\smtmap[x_1] = \langle c_1,v_1 \rangle \\ \smtmap[x_2] = \langle c_2,v_2 \rangle \\ c \textrm{\ fresh} \\ v \textrm{\ fresh}}
{\llparenthesis \mathrm{let\ } x = \textrm{assume}(x_1,x_2) \mathrm{\ in\ } C \rrparenthesis(\smtmap,\varphi)\ =\\ \llparenthesis C \rrparenthesis(\smtmap[x \mapsto \langle c, v\rangle], \varphi \land (c = c_1 \land c_2 \land v_1) \land (v = v_2))}\\

\inferrule[Nondet]
{\smtmap[x_1] = \langle c_1,v_1 \rangle \\ \smtmap[x_2] = \langle c_2,v_2 \rangle \\ c \textrm{\ fresh} \\ v \textrm{\ fresh}}
{\llparenthesis \mathrm{let\ } x = \textrm{nondet}(x_1,x_2) \mathrm{\ in\ } C \rrparenthesis(\smtmap,\varphi) = \\ \llparenthesis C \rrparenthesis(\smtmap[x \mapsto \langle c, v\rangle],\ \varphi \land \big( c \Rightarrow ((c_1 \land v = v_1) \lor (c_2 \land v = v_2)) \big) \land (c = c_1 \lor c_2))}\\

\inferrule[Mu]{\smtmap[x_0] = \langle c_0,v_0 \rangle \\ v \textrm{\ fresh}}
{\llparenthesis \mathrm{let\ } x = \muterm{x_s}{t_b}{x_e}{x_i} \textrm{\ in\ } C \rrparenthesis(\smtmap,\varphi) = \llparenthesis C \rrparenthesis(\smtmap[x \mapsto \langle c_0, v\rangle],\varphi)}\\

\end{mathpar}

\begin{theorem} Let $\langle \varphi, \smtmap \rangle = \llparenthesis C \rrparenthesis$. Let $\Gamma \in \llbracket C \rrbracket$. Let $x$ be a variable of $C$. Then, $x$ is in $\mathcal{M}$
and we choose $\langle c_x, v_x\rangle = \smtmap[x]$. Then we have:

\[ 
\left\{ \begin{array}{ll} \exists \Gamma \in \llbracket C \rrbracket(\varepsilon): \Gamma[x] = \bot &\quad\Rightarrow\quad \varphi \land \lnot c_x \textrm{\ is\ satisfiable}\\
                          \exists \Gamma \in \llbracket C \rrbracket(\varepsilon): \Gamma[x] = u \ne \bot &\quad\Rightarrow\quad \varphi \land c_x \land (v_x = u) \textrm{\ is\ satisfiable}\\
        \end{array} \right.
\]
\end{theorem}

\begin{proof}

The proof is by induction: at each step of the algorithm, the
$\varphi,\smtmap$ produced verify the property for the variables
already translated (which are in $\smtmap$). For each construct, we
show how a model of the formula $\mathscr{M}$ for $\varphi$ can be
extended to also satisfy the new constraints.

\end{proof}

\begin{corollary}
If $\varphi \land c_x \land (v_x = u) \textrm{\ is\ unsatisfiable}$, then $ \forall \Gamma \in \llbracket C \rrbracket(\varepsilon)$, we have $\Gamma[x] \ne u$
\end{corollary}

\begin{proof} This is just the contrapposite. \end{proof}

Thus in practice this translation allows to prove that a variable can
never be equal to some value. In particular in the case of boolean
values, it can be used to prove that some condition can never be
false, i.e. is always true; this allows to prove assertions about the
program.

Note that the translation is linear in the size of the term, notably
because we create new SMT variable for every LAF variable. Because the
translation is linear, we can see the translation to LAF term +
conversion to SMT formula as another implementation of the efficient
weakest precondition technique of Leino\cite{leino2005efficient}.

\subsection{Translation to Horn}
\label{sec:org386170d}

The translation is relatively similar, except that we make use of Horn
clauses to handle the recursion in the \(\mu\) term. The only real issue
is that in LAF term, the body of the \(\mu\) can use variables defined
outside of the body; this is called environment capture in functional
language. This is fixed by explicitly passing the contents of these
captured variables in the translation of the body.

\section{Translation from {\textsc{While}}}
\label{sec:orgb4a3cf1}
\label{sec:translation-from-while}

\newcommand{\transvar}[1]{\langle\!| #1 |\!\rangle}
\newcommand{\transexpr}[1]{(\!| #1 |\!)}
\newcommand{\transstmt}[1]{\{\!| #1 |\!\}}

The translation is relatively standard. The memory state corresponding
to each statement is represented by the tuple of the values of each
variable; \(\transvar{\cdot}\) is the constant mapping from variables to
indices. We have used "\([M \textrm{\ except\ } i \mapsto x]\)" as a
syntactic sugar for "\(\langle
M[0],M[1],\ldots,M[i-1],x,M[i+1],\ldots,M[n]\rangle\)". The translation
is quite naive; in particular the operations on tuples get/set, and
\emph{nondet} of tuples, should be simplified.

\begin{figure}[htbp]
\begin{align*}
\transexpr{\texttt{var}}(C,M) =
 &\ \langle C[\textrm{let\ } x' = M[\transvar{var}]], x' \rangle \\
 \transexpr{\textrm{op}_n(e_1,\ldots,e_n)}(C,M) =  
 &\ \textrm{let\ } \langle C_1,x_1 \rangle = \transexpr{e_1}(C,M) \textrm{\ in} \\
 &\ \textrm{let\ } \langle C_2,x_2 \rangle = \transexpr{e_2}(C_1,M_1) \textrm{\ in} \\[-5pt]
 &\ \vdots \\
 &\ \textrm{let\ } \langle C_n,x_n \rangle = \transexpr{e_n}(C_{n-1},M_{n-1}) \textrm{\ in} \\
 &\ \langle C[\textrm{let\ } x' = op_n(x_1,\ldots,x_n)], x' \rangle \\
 \transstmt{{\tt var} {\tt :=\ } e}(C,M) = 
 &\ \textrm{let\ } \langle C',x \rangle = \transexpr{e}(C,M) \textrm{\ in} \\
 &\ \langle C', [M \textrm{\ except\ } \transvar{\tt var} \mapsto x] \rangle \\
 \transstmt{s_1;s_2}(C,M) = 
 &\ \textrm{let\ } \langle C_1,M_1 \rangle = \transstmt{s_1}(C,M) \textrm{\ in} \\
 &\ \transstmt{s_2}(C_1,M_1) \\
 \transstmt{\textrm{\tt if\ } e \textrm{\tt \ then\ } s_{\textrm{then}} \textrm{\tt\ else\ } s_{\textrm{else}}}(C,M) = 
 &\ \textrm{let\ } \langle C_1,x \rangle = \transexpr{e}(C,M) \textrm{\ in} \\
 &\ \textrm{let\ } C_2 = C_1[\textrm{let\ } nx = \lnot x] \textrm{\ in} \\
 &\ \textrm{let\ } \langle C_3,M_{\textrm{then}} \rangle = \transstmt{s_{\textrm{then}}}(C_2,M) \textrm{\ in} \\
 &\ \textrm{let\ } \langle C_4,M_{\textrm{else}} \rangle = \transstmt{s_{\textrm{else}}}(C_3,M) \textrm{\ in} \\
 &\ \textrm{let\ } C_5 = C_4[\textrm{let\ } M_{\textrm{then}}' = \textrm{assume}(x,M_{\textrm{then}})] \textrm{\ in} \\
 &\ \textrm{let\ } C_6 = C_5[\textrm{let\ } M_{\textrm{else}}' = \textrm{assume}(nx,M_{\textrm{else}})] \textrm{\ in} \\
 &\ \textrm{let\ } C_7 = C_6[\textrm{let\ } M'= \textrm{nondet}(M_{\textrm{then}}',M_{\textrm{else}}')] \textrm{\ in} \\
 &\ \langle C_7, M' \rangle  \\
 \transstmt{\textrm{\tt while\ } e \textrm{\tt\ do\ } s \textrm{\tt\ done}}(C,M) = 
 &\ \textrm{let\ } {\bf M} = \textrm{fresh()} \textrm{\ in} \\
 &\ \textrm{let\ } \langle C_1,x \rangle = \transexpr{e}(C,{\bf M}) \textrm{\ in} \\
 &\ \textrm{let\ } C_2 = C_1[\textrm{let\ }M_2 = \textrm{assume}(x,{\bf M})] \textrm{\ in} \\
 &\ \textrm{let\ } \langle C_3,M_3 \rangle = \transstmt{s}(C_2,M_2) \textrm{\ in} \\
 &\ \textrm{let\ } C_4 = C[\textrm{let\ } M' = (\mu {\bf M}. C_3[M_3])(M) \textrm{\ in} \\
 &\ \textrm{let\ } \langle C_5, x' \rangle = \transexpr{e}(C_4,M') \textrm{\ in} \\
 &\ \textrm{let\ } C_6 = C_5[\textrm{let\ } nx' = \lnot x'] \textrm{\ in} \\
 &\ \textrm{let\ } C_7 = C_6[\textrm{let\ } M'' = \textrm{assume}(nx',M')] \textrm{\ in} \\
 &\ \langle C_7,M'' \rangle
\end{align*}
\caption{Translation from the \textsc{While} language to LAF terms}
\end{figure}

\begin{figure}[htbp]

\hspace{0.1\textwidth}\begin{minipage}{0.3\textwidth}

\begin{verbatim}
x := 0;
y := 0;
while(x < n) {
   x := x + 1;
   y := y + 1;
}
assert(x == y);
\end{verbatim}

\end{minipage}%
\begin{minipage}{0.6\textwidth}

{\footnotesize \begin{flalign*}
& \textcolor{white}{1}\hspace{0mm}   \textrm{let\ } x_0 \triangleq \textrm{unknown}_\mathbb{Z}() \textrm{\ in}\\[-3pt]
& \textcolor{white}{1}\hspace{0mm}   \textrm{let\ } y_0 \triangleq \textrm{unknown}_\mathbb{Z}() \textrm{\ in}\\[-3pt]
& \textcolor{white}{1}\hspace{0mm}   \textrm{let\ } n \triangleq \textrm{unknown}_\mathbb{Z}() \textrm{\ in}\\[-3pt]
& \textcolor{white}{1}\hspace{0mm}   \textrm{let\ } M_0 \triangleq \langle x_0, y_0, n_0 \rangle \textrm{\ in}\\[-3pt]
& \textcolor{white}{1}\hspace{0mm}   \textrm{let\ } M_1 \triangleq \langle 0, y_0, n_0 \rangle \textrm{\ in}\\[-3pt]
& \textcolor{white}{1}\hspace{0mm}   \textrm{let\ } M_2 \triangleq \langle 0, 0, n_0 \rangle \textrm{\ in}\\[-3pt]
& \textcolor{white}{1}\hspace{0mm}   \textrm{let\ } \langle x_3, y_3 \rangle = (\mu{\bf M}. \\[-3pt]
& \textcolor{white}{1}\hspace{0mm}   \textrm{let\ } x_1 \triangleq M[0] \textrm{\ in}\\[-3pt]
& \textcolor{white}{1}\hspace{8mm}   \textrm{let\ } c_1 \triangleq x < n \textrm{\ in} \\[-3pt]
& \textcolor{white}{1}\hspace{8mm}   \textrm{let\ } M_3 \triangleq \textrm{assume}(c,M) \textrm{\ in}\\[-3pt]
& \textcolor{white}{1}\hspace{8mm}   \textrm{let\ } x_1 \triangleq M_3[0] \textrm{\ in}\\[-3pt]
& \textcolor{white}{1}\hspace{8mm}   \textrm{let\ } y_1 \triangleq M_3[1] \textrm{\ in}\\[-3pt]
& \textcolor{white}{1}\hspace{8mm}   \textrm{let\ } n_1 \triangleq M_3[2] \textrm{\ in}\\[-3pt]
& \textcolor{white}{1}\hspace{8mm}   \textrm{let\ } x_2 \triangleq x_1 + 1 \textrm{\ in}\\[-3pt]
& \textcolor{white}{1}\hspace{8mm}   \textrm{let\ } y_2 \triangleq y_1 + 1 \textrm{\ in}\\[-3pt]
& \textcolor{white}{1}\hspace{8mm}   \textrm{let\ } M_4 \triangleq \langle x_2, y_2,n_1 \rangle \textrm{\ in}\\[-3pt]
& \textcolor{white}{1}\hspace{8mm}   M_4)(M_2) \textrm{\ in}\\[-3pt]
& \textcolor{white}{1}\hspace{0mm}   \textrm{let\ } x_3 \triangleq M_4[0] \textrm{\ in}\\[-3pt]
& \textcolor{white}{1}\hspace{0mm}   \textrm{let\ } c_2 \triangleq x_3 < n \textrm{\ in} \\[-3pt]
& \textcolor{white}{1}\hspace{0mm}   \textrm{let\ } nc \triangleq \lnot c_2\textrm{\ in} \\[-3pt]
& \textcolor{white}{1}\hspace{0mm}   \textrm{let\ } M_5 \triangleq \textrm{assume}(nc,M_4) \textrm{\ in}\\[-3pt]
& \textcolor{white}{1}\hspace{0mm}   \textrm{let\ } x_6 \triangleq M_5[0] \textrm{\ in}\\[-3pt]
& \textcolor{white}{1}\hspace{0mm}   \textrm{let\ } y_6 \triangleq M_5[1] \textrm{\ in}\\[-3pt]
& \textcolor{white}{1}\hspace{0mm}   \textrm{let\ } c_3 \triangleq x_6 = y_6 \textrm{\ in} \\[-3pt]
& \textcolor{white}{1}\hspace{0mm}   c
\end{flalign*}}

\end{minipage}

\vspace{-3mm}
\caption{Full (naive) translation of the assertion in the program on the left. A better translation would 
remove useless terms such as $M_1$ and $x_0$, and would realize that $n$ is not modified in the loop (and thus does not need to be in the loop tuple).}
\label{fig:trans-exemple}
\end{figure}
\end{document}